\documentclass[preprint]{iacrtrans}
\usepackage[utf8]{inputenc}

\usepackage{booktabs}           

\usepackage{url}
\usepackage{subfigure}
\usepackage{alltt}
\usepackage{fancyhdr}
\usepackage{stmaryrd} 
\usepackage{newfloat} 
\usepackage{framed}   
\usepackage{cite}
\usepackage{graphicx}
\usepackage{textcomp}
\usepackage{xcolor}
\usepackage{dsfont} 
\usepackage{upgreek}   
\usepackage{enumitem}  
\usepackage{marvosym}  



\makeatletter
\@ifpackageloaded{txfonts}\@tempswafalse\@tempswatrue
\if@tempswa
  \DeclareFontFamily{U}{txsymbols}{}
  \DeclareFontFamily{U}{txAMSb}{}
  \DeclareSymbolFont{txsymbols}{OMS}{txsy}{m}{n}
  \SetSymbolFont{txsymbols}{bold}{OMS}{txsy}{bx}{n}
  \DeclareFontSubstitution{OMS}{txsy}{m}{n}
  \DeclareSymbolFont{txAMSb}{U}{txsyb}{m}{n}
  \SetSymbolFont{txAMSb}{bold}{U}{txsyb}{bx}{n}
  \DeclareFontSubstitution{U}{txsyb}{m}{n}
  \DeclareMathSymbol{\aleph}{\mathord}{txsymbols}{64}
  \DeclareMathSymbol{\beth}{\mathord}{txAMSb}{105}
  \DeclareMathSymbol{\gimel}{\mathord}{txAMSb}{106}
  \DeclareMathSymbol{\daleth}{\mathord}{txAMSb}{107}
\fi
\makeatother

\def\authorone{Peter T. Breuer}

\def\E{{\cal{E}}}
\def\enc[#1]{\E[#1]}
\def\D{{\cal{D}}}
\def\dec[#1]{\D[#1]}
\def\QED{\raisebox{0.8ex}{\framebox{\kern0.2ex}}}
\def\C#1{{\mathbb C}[#1]}
\def\Ss#1{{\mathbb C}^L[#1]}

\def\iff{\mathop{\rm~iff~}}
\def\H{{\rm H}}

\DeclareFontShape{OT1}{cmtt}{bx}{n}
     {
      <5> <6> <7> <8> <9>
      <10> <10.95> <12> <14.4> <17.28> <20.74> <24.88> cmbtt10
      }{}
\DeclareFontShape{OT1}{cmtt}{b}{n}
  {<->sub * cmtt/bx/n}{}

\newtheorem{fact}{Fact}

\makeatletter
\def\blfootnote{\xdef\@thefnmark{}\@footnotetext}
\makeatother

\makeatletter
\DeclareRobustCommand*\cal{\@fontswitch\relax\mathcal}
\makeatother

\DeclareFloatingEnvironment[fileext=box,placement={!tb},name=Box]{mybox}

\newenvironment{frameenv}[1][]
{%
 \begin{mybox}[#1]%
 \begin{framed}
 \begin{minipage}{0.99\columnwidth}
 \begin{footnotesize}
}{%
 \end{footnotesize}\end{minipage}\end{framed}\vspace{-2ex}\end{mybox}%
}

\newcounter{myfootnote}
\newcommand{\myfootnote}[2]{%
\let\oldthefootnote=\thefootnote%
\renewcommand{\thefootnote}{#1}%
\footnote{#2}%
\addtocounter{myfootnote}{1}%
\let\thefootnote=\oldthefootnote%
}

\makeatletter
\newcommand{\mypara}[1]{%
\par\noindent%
\refstepcounter{paragraph}%
\S\thesection.\arabic{paragraph}
\textbf{#1}\enspace\ignorespaces
}
\makeatother

\makeatletter
\newcommand{\customlabel}[2]{%
\protected@write\@auxout{}{\string\newlabel{#1}{{#2}{\thepage}{}{section.\thesection}{}}}%
}
\makeatother

\makeatletter
\def\subsection{\@startsection{subsection}{2}{\z@}%
                                     {-3.25ex\@plus -1ex \@minus -.2ex}%
                                     {1.5ex \@plus .2ex}%
                                     {\reset@font\normalsize\bfseries}}
\makeatother

\begin{document}

\date{}

\title{
Chaotic Compilation for Encrypted Computing:\\
Obfuscation but Not in Name
}
\titlerunning{Obfuscation but Not in Name}
\authorrunning{P.\,T.~Breuer}

\author{
\authorone
} 
\institute{
Hecusys LLC, GA, USA,
\email{ptb@hecusys.com}
}

\maketitle

\widowpenalty=0 
\clubpenalty=0  

\begin{abstract}
An `obfuscation' for encrypted computing is quantified
exactly here, leading to an argument that security against
polynomial-time attacks has been achieved for user data via the
deliberately `chaotic' compilation required for security properties
in that environment.  Encrypted computing is the emerging science
and technology of processors that take encrypted inputs to encrypted
outputs via encrypted intermediate values (at nearly conventional
speeds).  The aim is to make user data in general-purpose computing
secure against the operator and operating system as potential
adversaries.  A stumbling block has always been that memory addresses
are data and good encryption means the encrypted value varies
randomly, and that makes hitting any target in memory problematic without
address decryption, yet decryption anywhere on the memory path would
open up many easily exploitable vulnerabilities.  This paper `solves
(chaotic) compilation' for processors without address decryption,
covering all of ANSI C while satisfying the required security properties
and opening up the field for the standard software tool-chain
and infrastructure.  That produces the argument referred to above,
which may also hold without encryption.
\end{abstract}

\section{{Introduction}}
\label{s:Intro}

\noindent
This article explains recent advances in understanding and practice in
the emerging technology of {\em encrypted computing}
\cite{fletcher2012,BB13a,oic,heroic,BB16b,BB18a}, in particular
`chaotic compilation' for this context \cite{BB17a}, which 
supports the security proofs and now has broken through to allow essentially
all extant ANSI C \cite{ansi99} source codes to be compiled (the known lacks are
longjmp/setjmp), nearly as they stand (strict typing is necessary),
opening up the field. An unexpected by-product is an argument that
obfuscation safe against polynomial time attacks is produced by
this kind of compilation, perhaps even in the absence of encryption.

Encrypted computing means running on a processor that takes encrypted
inputs to encrypted outputs via encrypted intermediate values.
The operator and operating system are the potential adversaries
in this context. Encrypted computing  aims to:
\begin{quote}
\em
Protect user data from the powerful operator of the system.
\end{quote}
A subverted operating system is as much `the operator' as is a human
with administrative privileges, perhaps obtained by interfering with the
boot process, and this document will use `the operator' indiscriminately
for both.  From a systems perspective, the operator (or operating
system) are merely the operator mode of the processor.  A processor
starts in operator mode when it is switched on, in order to load
operating system code into reserved areas of memory from disk. 
Conventional application software relies on the processor to change from
user mode to operator mode and back again for system support (e.g.,
disk I/O) as required, so the operator mode of the processor presents
difficulties as adversary.  Security is possible only because `an honest
operator should not care what user data means' -- informally, the
operator/system is there to change the tapes when the machine beeps.

The reader should take onward with them from the outset that a processor
for encrypted computing works encrypted only in user mode.  In operator
mode it works unencrypted (with unrestricted access to registers and
memory) entirely as usual. A good mnemonic is:
\begin{quote}
\em
Privilege means no encryption,
no privilege means encryption.
\end{quote}
Any user data encountered while in operator mode (and also in user mode) will
be in encrypted form so the operator's privilege of unrestricted access
does not necessarily imply understanding of what user data means, or an ability to
alter it meaningfully. It can always be copied for later analysis, or
overwritten with zeros, one word substituted with another, etc, but
security proofs \cite{BB18c} show that cannot amount to more than
guesswork.

A fundamental technical motivation in this area is the intuition that a
processor that `works encrypted' must inherently be not less secure
than one that does not, and there are reasons why
it should not be (much) slower that have been borne out in hardware,
as described below.
The non-technical motivations are manifold and should not need enumeration.
One plausible scenario for
an attack by the operator is where cinematographic animation is rendered
in a server farm.  The computer operators have an opportunity to pirate
for profit parts of the movie before release and they may be tempted.
Another possible scenario is where a specialized facility processes
satellite photos of a foreign power's military installations to spot
changes.  If an operator (or hacked operating system) can modify the
data to show no change where there has been some, then that is an option
for espionage.  Informally:

\begin{definition}[Successful attack]
\em A {\em successful attack}
by the operator is one that discovers the plaintext of user data or
alters it to order.
\end{definition}

\noindent
That is meant statistically, so an attack that succeeds
more often than chance counts.

Overwriting encrypted data with zeros does not count because that does
not set the plaintext under the encryption to a value that the adversary
can predict, even stochastically, supposing the encryption itself is
secure.  The operator can always act chaotically or destructively, by
rewriting memory randomly or turning the machine off, for example.

Pointers (and arrays) are the major obstacle to compiling real,
extant, source codes in the encrypted computing setting.  Pointers
underly any practical
programming language: in C (and C++) they are explicit, but in other
languages they are implicit.  In Fortran\,95 all variables are implicitly
pointers, so passing the variable to a subroutine allows the caller's
variable to be changed by the callee.  Passing an array allows entries
in the caller's array to be rewritten in most languages.  In Java every
object is a pointer, and a new object instead of another pointer to the
same object must be created via {\bf clone}().

`The problem with pointers' is that (1) the memory unit will receive an
encrypted address (`pointer') and (2) it is not privy to the encryption.
The rationale for (1) is that while the processor design could
accommodate address decryption on the memory path, that opens up an attack
vector in which the operator passes encrypted data as an address in a
load or store machine code instruction and checks (physically or
programmatically) where in memory it accesses, thus decrypting the
data, so it is not a good idea for security.  It is also faster for
hardware not to have encryption or decryption units on the memory path.
The rationale for (2) is an attacker with physical access could walk away with
the RAM chips and they might contain plaintext or keys if the memory
unit could do decryption itself (c.f., {\em cold boot} attacks
\cite{halderman2009lest,Gruhn2013}, physically freezing RAM to keep data
while the chips are moved), or if the processor hardware decrypts for
it. 

The conclusion to draw is that memory should do without decryption both
internally and externally.  But then it will receive some ciphertexts
that are each a different encryption of the same programmed address and
it cannot tell they are intended to mean the same thing: the programmer
says to write $x$ at address $123$ and the encrypted value 0x123456789A
of $123$ is passed to the memory as the ciphertext for the address at
runtime; if the programmer aspires to read it back via address $123$
then a ciphertext 0xA987654321 that is an alternative encryption of
$123$ may be passed as address, and content $y$ different to $x$ is
returned from there. One can get around that in hand-crafted fashion
by specifying exactly every time what encrypted address is to be used,
but that only works for simple programs -- it will break down for 
programs that use arrays of pointers to other arrays of pointers
to objects which may contain pointers.  In other words, programs will
break, often silently.  Strict typing can limit programs to those that
will work in the setting, but the programming experience would be
frustrating and porting existing programs to the platform would require
a prohibitively high degree of semantic understanding of the original.

We have now learned how to work around this hardware semantics
problem via software logic, while maintaining the security properties
required. That opens up encrypted computing for a
traditional software development cycle involving high-level source
language programming and a compile, assemble, link, load tool-chain with
operating system support.

This article is organized as follows.
Section~\ref{s:Ref} will set out bullet points about
encrypted computing platforms to refer to in the rest of the text. 
The reader should defer to those points.  For example, it is
not the case that the operator has access to the encryption key, as a
reader may suppose from experience with conventional platforms, and the
first point emphasizes that.  Nor is encryption done in software, as
the reader may have experienced, so it
does not take different times for different data.  Nor do instructions
aborted while still in the processor pipeline leave a mark in cache, nor
may they otherwise be detected, etc. There is nearly a decade of
research in the field, and there are no easily picked openings.
Section~\ref{s:Key} expresses abstract terms that
hopefully resonate with cryptologists, but ready-made
matches for this new field are scarce and the reader should first and
foremost take the concepts defined there on their own terms.
Section~\ref{s:Back} describes existing platforms for encrypted
computing and discusses the theory.
Section~\ref{s:FxA} resumes a modified Open\-RISC
(\url{http://openrisc.io}) machine code instruction set for encrypted
computing described in \cite{BB17a}, needed for security.
Section~\ref{s:Comp} first resumes the basis of chaotic 
compilation restricted to integers and call by value
a la \cite{BB17a} and then covers ramified atomic types
(long integers,  floats,  doubles, etc.),
arrays and pointers (hence call by reference), `struct' (record) and
union types.  Theory is developed in Section~\ref{s:Theory}. It
quantifies the entropy in a runtime trace for code compiled following
\eqref{e:maxH}, characterizing that as `best possible'.
Section~\ref{s:Discuss} discusses further implications for security,
including the argument that this setting has security against polynomial
time (in the word-length, $n$) attacks on user data, with or without
encryption.

\section{Reference Points}
\label{s:Ref}

\noindent
This section sets out touchstones on encrypted computing for
the reader to refer to.

\mypara{Encryption key access} is impossible programmatically. The key
is embedded in the processor hardware and there are no instructions
that read it.
Keys are either installed at manufacture, as with Smartcards
\cite{SmartCard}, or uploaded securely in public view via a
Diffie-Hellman hardware circuit \cite{buer2006cmos} to a write-only and
otherwise inaccessible internal store.
The privy user (which term does not include the operator) can always interpret
the encrypted data, elsewhere, for safety, as they know what the key is.

\mypara{Prototype processors}
for encrypted computing at near conventional speeds
already exist (see Section~\ref{s:Back}) and have existed since a first
recognizable try at the idea \cite{fletcher2012} about 2012.  Before that,
computing {\em systems} (not processors) that operated partially
encrypted were not uncommon for commercial movie and music reproduction
going back as far as 1993 \cite{hartman1993system} and 2000
\cite{hashimoto2001}, but their processor component operated in the
ordinary way while memory (RAM) or disk or other peripherals stored in
encrypted form the media to be safeguarded.  Modern Intel
SGX\texttrademark{} systems
\cite{Anati2013,hoekstra2013using,mckeen2013innovative} are not
conceptually different in that regard, but with extra sandbox technology
(memory enclaves into which the programmer may voluntarily consign
their code, with separate caches and register sets).  The idea in
encrypted computing is to rely instead on encryption through-and-through
as the security barrier protecting user data.  A basic question to
answer is whether encrypted computing potentially reduces the security
of the encryption (`no', per \cite{BB18c}).

\mypara{Platform/hardware issues} such as the real
randomness of random numbers or power side-channel information leaks
\cite{Zhang2012} will not be at issue here. Existing technological
defenses for conventional processors \cite{kissell2006method} may be
applied in practice, and are to be supposed.

\mypara{The intended mode of working} is `remote processing':
\begin{quote}
\em
The user compiles program and data,  sends both
away, and gets back output.
\end{quote}
\noindent
New data means a new encoding and compilation.  The word `encoding' is
used because more than encryption (and more than
vanilla compilation) is involved.  The compilation is `chaotic', as
explained below.  That involves producing both a new encoding for the
data throughout the program and a new machine code program that can cope
with it and then encoding and encrypting the data and
partially encrypting the machine code (only constants in
the instructions are encrypted - the `opcode' and other fields are
plaintext).

Continuous running may become feasible as technology
for dynamic code update \cite{Hicks2005} in this setting develops, but
for the purposes of this paper a mode of use other than the above is not
contemplated.

\mypara{Key management} is not an issue via this
argument: if (a) user B's key is still loaded when user A runs, then A's
programs do not run correctly as the encryption is wrong for them; if
(b) B's key is in the machine together with B's program when A runs,
then user A cannot supply appropriate encrypted inputs nor interpret the
encrypted output. 

\mypara{Security} user on user boils down in this setting to security
for the user against the operator as the most
powerful potential adversary in the system, and
it is proved in \cite{BB18c} (the argument is reprised in the Appendix;
Lemma~\ref{at:4}), that

\begin{enumerate}[noitemsep,label=(\roman*)]
\item a processor that supports encrypted
computing,
\item an appropriate machine code instruction set
architecture,
\item a `chaotic' compiler as described below
\end{enumerate}
together provide a property that parallels for this setting
classic {\em cryptographic semantic security} \cite{Goldwasser1982} (CSS) for
encryption, better known via the semantically equivalent {\em ciphertext
indistinguishability under chosen-plaintext attack} (IND-CPA).  The
latter means there is no (polynomial time) method for an adversary to
tell which is which between the ciphertexts of two plaintexts
(stipulated by the adversary, both the same length), to any degree
significantly above chance (as the key/block size tends to infinity).

\mypara{An appropriate machine code instruction set} as referred to
above is one that satisfies four axioms introduced in \cite{BB16b} and
set out in Box~\ref{b:2} here.  They are essentially (a) atomicity, (b)
encrypted working, (c) malleability, and (d) absence of collisions
between the ciphertext constants in program instructions and data at
runtime.

\mypara{Source code language coverage}
for encrypted computing that extended to (32-bit plaintext) integers,
arithmetic, conditionals and call-by-value was first achieved in
\cite{BB17a} and is explained in Section~\ref{s:Comp}.  It is
extended here to pointers and all of {\sc ansi} C \cite{ansi99} (except
setjmp/longjmp), with its long long, float, double, array, struct
(record) and union data types.  That `solves' the practical problem of
compilation for all encrypted computing.

\mypara{Primitive operations} in the processor will be taken to include all
encrypted 32-bit integer arithmetic.  That means for example that
encrypted data can be `added' by the hardware in
one operation, via an appropriate machine code instruction, producing an
encryption of $x+y$ (mod $2^{32}$) as result from encryptions of $x$ and
$y$ as operands.  Similarly for the other primitives of the computer
arithmetic, which are specified exactly in Section~\ref{s:FxA}.  
Careful specification is necessary because it turns out that the
standard arithmetic operations in their conventional form are dangerous
to security \cite{Rass2016}. 

Since hardware is not the focus of this paper, encrypted 64-bit integer
arithmetic will also be taken as primitive.  It is
carried out on two encrypted 32-bit integers representing the high and
low bits respectively (software subroutines for each operation 
using the 32-bit instruction set only is an alternative).
Encrypted (32-bit and 64-bit) floating point arithmetic will
only be treated in the Appendix but should also  be taken as
primitive.
These primitives are each supported by at least one of the prototype processors
discussed in Section~\ref{s:Back}.

\section*{{Notation}}
\label{s:Not}

\noindent
Encryption (with key $K$ understood) will be denoted $x^\E$ or $\E[x]$ of
plaintext value $x$ and should be read as a multi-valued function of $x$,
i.e., a single-valued function $\E[p\mathop{\cdotp} x]$ when the hidden
padding $p$ is taken into account.
Decryption is
$\zeta^\D=\D[\zeta]$, with $x=\D[x^\E]$.  The encryption aliases of a
ciphertext $\zeta_0$ are those $\zeta_1$ with the same plaintext as it
under decryption $\D$, i.e.,
$\zeta_0 \mathop{\mathop\equiv\limits_{\D}} \zeta_1$, meaning $\D[\zeta_0] =
\D[\zeta_1]$. To avoid excess notation that equivalence is written as
{\em equality} on the cipherspace, so equal aliases $\zeta_0 = \zeta_1
$, i.e., $\zeta_0 \mathop{\mathop\equiv\limits_{\D}} \zeta_1 $,
may be non-identical, with $\zeta_0\mathop{\ne}\limits_{\rm id}
\zeta_1$. They are encryption alternatives or {\em aliases} of the same
plaintext $x=\zeta_0^\D = \zeta_1^\D$.  Then $x^\E$, $\E[x]$ 
denote particular but unspecified aliases, and
$\zeta = \E[\D[\zeta]]$ is true (recall that equality on the cipherspace
is equivalence under decryption).  Key $K$ may be mentioned as
an extra parameter in $\E[K,x]$ and $\D[K,\zeta]$.  Padding
varies the value but not the $\mathop\equiv\limits_{\D}$ equivalence
class, which is the equality, with
$\E[p\mathop{\cdotp} x]=\E[q\mathop{\cdotp} x]\mathop{\ne}\limits_{\rm
id}\E[p\mathop{\cdotp} x]$ for $p\ne q$, and $\E[p\mathop{\cdotp}
x]\ne\E[p\mathop{\cdotp} y]$ for $\E[x]\ne\E[y]$.

The operation on the ciphertext domain corresponding to $o$ on the
plaintext domain will be written $[o]$, where
$x^\E\mathop{[o]}y^\E=\E[x\mathop{o} y]$. The relation on the
ciphertext domain corresponding to $R$ on the plaintext domain will be
written $[R]$, where $x^\E\mathop{[R]}y^\E \iff x\mathop{R} y$.
These are well-defined with respect to the cipherspace equality.

\section{Key Security Concepts}
\label{s:Key}

The game-theoretic formulation of the classic IND-CPA security
definition for this setting\,is:
\begin{enumerate}[noitemsep,label=G1.\arabic*]
\item the operator selects any program $p$ and input data $d$;
\item the user compiles it (twice) and shows the operator the
two plaintext codes $p_1$, $p_2$ and input data $d_1$, $d_2$;
\item the user encrypts those to 
$p_1'$, $p_2'$ and  $d_1'$, $d_2'$ and passes those to the
encrypted computing platform;
\item the operator examines the running
codes $p_1'(d_1')$, $p_2'(d_2')$,  interferes, experiments,
observes intermediate and final results $e_1'$, $e_2'$
(all in encrypted form);
\item the operator attempts to say which is which.
\end{enumerate}
The argument of Lemma~\ref{at:4} shows that in an encrypted computing
context there is no method (in particular no polynomial time one) the
operator can use to be right more often than if they were trying plain
IND-CPA against the encryption.  The probability of success given that
there is no advantage against the encryption alone is not different from
chance, i.e.  1/2 (exactly).

\begin{definition}
That the operator cannot win the game above with different than 1/2
probability (exactly, for all lengths of the encrypted word on the
platform) for this context, given that the encryption itself is
independently secure against IND-CPA, will be called {\em cryptographic
semantic security {\em relative} to the security of the encryption}
($\rho$CSS), for user data against operator mode as adversary.
\end{definition}

The rationale behind it is discussed below, but it is clear that it
holds for certain programs $p$: those that have no instructions and thus
do nothing to data, passing input to output directly, without any
interference.  Given that the encryption is secure, an adversary can do
nothing but guess blindly as to what the data is.

An interesting corollary (Remark~2 of \cite{BB18c}) has it that
there is an encrypted program that does decryption (it is the decryption
algorithm for the encryption in use in the processor, compiled for the
encrypted computing environment), but the operator cannot 
build it correctly, neither from scratch nor out of scavenged parts from
encrypted programs, with any probability of success above random chance.
That is so though they know what the encryption is (not the key)
and the canonical structure of a program for decryption.

The property $\rho$CSS is above all a statement about whether encrypted
computing
provides additional information to the operator that might lessen the
security of the encryption (`no').  For example, it might be possible
for the operator to recognize a computation $2+2=2*2$ (encrypted), via
coincidence of the encrypted operands and results.  The operator might
try inserting the same encrypted value for (supposedly) $2$ at all the
places in that calculation and see the encrypted answers are identical.
The argument of Lemma~\ref{at:4} says that no, even though the operator
thinks they have identified the situation, in fact other interpretations
of the operands and results that are consistent with the operational
semantics of the processor are possible, and all interpretations are
equally possible given the `right' compiler as generator.  That is a
perspective from mathematical logic and model theory, and  it may be
expressed for a general reader as follows: if one takes the axioms
like $x+y=y+x$ that describe the specially designed primitive operations
of the processor (Table~\ref{tb:1}), and chooses facts $r_1=2$, $r_2=4$
etc.\ consistent with the axioms for the unknown values $r_1$, \dots
$r_n$ at $n$ chosen points in the runtime trace of a program, then there
are also other, different but compatible choices such as $r_1=5$,
$r_2=6$ etc.\ that are equally plausible.  That holds for $n=1$ chosen
point, wherever it is (Corollary~2 of \cite{BB18c}), and in general
independently for any $n$ chosen points, with certain exceptions
characterized in this paper (for example, measuring the input $r_1$ and
output $r_2$ of a copy instruction does not allow $r_1$ and $r_2$ to be
chosen independently -- one of them may be chosen independently, and
then the other must be equal to it).  An informal rendering of $\rho$CSS
might be that good security in encrypted computing boils down to good
encryption.

A similar, but different, game is:
\begin{enumerate}[noitemsep,label=G2.\arabic*]
\item the operator selects any program $p$ and input data $d$ and
selects a polynomially defined point $r_n$ in the trace (e.g., the
last point at which register $r$ changes in the first $n^3$ steps);
\item the user compiles it for a
platform that does $n$-bit computing beneath the encryption,
forming $p_n$ and $d_n$, which are not shared with the operator;
\item the user encrypts to 
$p_n'$ and  $d_n'$ and passes those to the encrypted computing platform;
\item the operator examines the running
code $p_n'(d_n')$, interferes, experiments,
observes intermediate and final results (all in encrypted form);
\item the operator tries to say what the plaintext value beneath the
encryption in $r_n$ is.
\end{enumerate}
The argument in Section~\ref{s:Discuss} says that if the
encryption is secure for IND-CPA, then the operator cannot win this
game with a probabilty above chance, as $n\to\infty$ 
(but what the probability is for a given program
$p$ and word length $n$ is not going to be known).

If the operator could win this game, then they could win the first game
given at the front of this section (for sufficiently large $n$), so it
ought not to happen.  But the argument is independent and also suggests that
encryption may not be necessary.

Chaotic compilation for encrypted
computing is as follows. It is related to `obfuscation' in that
it makes code harder to read than it would otherwise be, but
`obfuscation' in the security
field refers to transforming one source code to another, while
compilation transforms source code to object code
(assembler, or machine code), not source code.
so a compiler cannot be a classical {\em obfuscator}. Instead,
chaotic compilation has properties required to prove $\rho$CSS.
It behaves stochastically:  applied twice to the same
source code, it emits two (probably) different object codes.
Chaotic compilation always at least:
\begin{enumerate}[label={\protect{(\Alph*)}},ref=\protect{\Alph*},noitemsep]
\item {\em produces identically structured machine codes;}
\label{p:1a}
\item {\em produces identically structured runtime traces;}
\label{p:1b}
\item {\em varies only encrypted constants in the code;}
\label{p:1c}
\item {\em varies runtime data beneath the encryption}
\label{p:1d}
\end{enumerate}
across different recompilations of the same source code.
The aim of (\ref{p:1a}-\ref{p:1d}) is that:
\begin{quote}
\em
The object codes and their traces look the same to
an ignorant observer.
\end{quote}
That is an observer ignorant of the encryption key and hence putatively
unable to read the plaintext of encrypted words.
The same instructions (modulo different embedded encrypted constants)
will run in the same order, with the same branches and
loops, but runtime data beneath the encryption differs following a
scheme known only to the user.  The math and computer science behind that
is reprised in Section~\ref{s:Back}. 

Chaotic compilation in particular means producing
code within the framework of (\ref{p:1a}-\ref{p:1d})
such that an adversary cannot count
on 0,\,1,\,2,\,etc.\ occurring frequently
beneath the encryption in a program trace,  which is naturally
the case in
a program written by a human. That would enable
statistically-based dictionary attacks \cite{katz1996}
against the encryption.  The desired property is:
\begin{definition}[Chaotic compilation]
No data value beneath the encryption
appears at runtime with higher probability than
another.
\label{p:obfusc}
\end{definition}
\noindent
That is measured across recompilations, which are generated 
stochastically. There is no dependence on the key/block size $n$ in
that definition (which also makes it different in kind from classical
definitions of obfuscation). The principle applies both for observations of
single words and also for simultaneous (`vector') observations at multiple
points in a trace, but that is much harder to achieve.
This paper proves it for single words and quantifies it for
vectors.  Note that the property is literally violated, e.g,  in
implementations \cite{DGHV10} of fully homomorphic encryptions
\cite{RAD78,Gentry09} (FHE), where the output of a 1-bit AND
(multiplication) operation is 0 beneath the encryption with 
probability $\tfrac{3}{4}$ (Box~\ref{b:1} (a)).\footnote{That 0 is a
probable outcome from multiplication in a FHE $\E$ is not an extra
liability because in 1-bit arithmetic $\E[x]+\E[x]=\E[0]$ with certainty
from any observed encrypted value $\E[x]$.  It can also be relied on
that $\E[1]$ is one of the inputs in any nontrivial calculation because
`all-zeros' as inputs propagates through to all-zeros as output via
$\E[0]+\E[0]=\E[0]*\E[0]=\E[0]$.  }

\begin{frameenv}[t]
\begin{flushleft}
\small
\refstepcounter{mybox}
\centerline{\rm Box \Roman{mybox}}
\medskip
\begin{minipage}[t]{0.42\columnwidth}
{\footnotesize
(a) A fully homomorphic encryption (FHE) $\E$ of 1-bit data
lacks the cryptographic semantic security (CSS) property.}
\customlabel{b:1}{\Roman{mybox}}
\begin{align*}
\E[0]*\E[0]&=\E[0]\\
\E[0]*\E[1]&=\E[0]\\
\E[1]*\E[0]&=\E[0]\\
\E[1]*\E[1]&=\E[1]
\end{align*}
\noindent
Guessing 0 as outcome is right 75\% of the time.
\end{minipage}
\hfill
\begin{minipage}[t]{0.42\columnwidth}
(b) A FHE program that adds 2-bit data to itself:

\begin{align*}
\E[0]+\E[0]&=\E[0]\\
\E[1]+\E[1]&=\E[2]\\
\E[2]+\E[2]&=\E[0]\\
\E[3]+\E[3]&=\E[2]
\end{align*}
has output $y=2x$ that is 100\% even, breaking CSS.
\end{minipage}
\end{flushleft}
\end{frameenv}

What makes it hard to achieve Defn.~\ref{p:obfusc} 
simultaneously for many different observed
values at different points in the trace
is that computational semantics has rules to it.
Apart from (\ref{p:1a}-\ref{p:1d}),
the `chaotic compiler' is obliged to generate machine code in
which:
\begin{enumerate}[label=\protect{(\Alph*)},ref=\protect{\Alph*},start=5,noitemsep]
\item {\em a copy instruction preserves data exactly;}
\label{p:5a}
\item {\em the variations introduced by the compiler are equal
      where any two control paths join.}
\label{p:5b}
\end{enumerate}
Any observer can deduce data beneath the encryption is copied in
\eqref{p:5a}, because the same ciphertext is seen.
The condition \eqref{p:5b} refers to the end of a loop, after
conditional blocks, at subroutine returns, at the label target of a {\bf
goto}.  That it holds is deducible by an adversary from
basic computational principles: since the number of times through a
loop is generally not predictable at compile-time, the compiler must ensure
the same conditions are re-established at the end of each loop as
prevail at the start, ready for another go-through.
\def\Obj{$\mathcal{S}$}

The constraints (\ref{p:5a}-\ref{p:5b}) impose an underlying order on what
is desired to be an apparently chaotic scheme of variations from nominal
in the plaintext data beneath the encryption in a program trace.
An overall {\em strategy} for achieving Defn.~\ref{p:obfusc})
is as follows \eqref{e:S}:
\begin{equation}
\parbox{0.85\columnwidth}%
{\em
Vary a machine code program's instruction constants so every possible data
variation beneath the en\-cryption is obtained with {\em uniform}
probability.
}%
\tag{\Obj}%
\label{e:S}%
\end{equation}

\def\maxH{$\mathcal{T}$}

\noindent
The probability is taken across traces, one for each
recompilation of the same source code. Unfortunately, modifications
require knowledge of the programmer's intention
as expressed in the source code, and reverse engineering that from
machine code is in general a known Turing Halting Problem equivalent
(i.e., computationally impossible).  It is up to the compiler to provide
variation, working stochastically from source code. However,
the `uniform \dots across the range' is a key  that indicates how the
compiler ought to generate code that boils down to a particular tactic
as explained below.

The compiler {\em tactic per increment of code} to implement the
strategy \eqref{e:S} is as follows:
\begin{equation}
\parbox{0.85\columnwidth}{
\em Every machine code instruction that writes should introduce maximal
entropy.}
\tag{\maxH}
\label{e:maxH}
\end{equation}
That refers to entropy introduced into the runtime trace by the
compiler's variations.

What this means is that
the compiler must exercise fully the possibilities for varying the trace at
every  opportunity in a compiled program.  The only way to vary the
trace 
is via an instruction that writes something, otherwise there is no
effect.
`Writes' means any change that is testable, even if not observed directly
--
the outcome of a comparison instruction cannot be read from 
registers, for example, but it is tested by where the branch instruction
jumps to, so it counts as a `write'.
The compiler should not, for example, always use 1 in an addition
instruction if the possibility exists of using a  different number and
still satisfying the programmer's intention.

\section{Background and Related Work}
\label{s:Back}

Several fast processors for encrypted computing are described in
\cite{BB18b}.  Those include the 32-bit KPU \cite{BB16b} with 128-bit
AES encryption \cite{DR2002}, which on the industry-standard Dhrystone
\cite{Weicker84} v2 benchmark is reported in \cite{BB18b} to
run encrypted at the speed of a 433\,MHz classic Pentium (tables
equating different processors are at
\url{www.roylongbottom.org.uk/dhrystone%20results.htm}) with a 1\,GHz
clock speed, and the older 16-bit HEROIC \cite{oic,heroic} with
2048-bit one-to-one Paillier encryption \cite{Pail99}, which runs like a
25\,KHz Pentium, as well as the recently announced 32-bit CryptoBlaze
\cite{cryptoblaze18} with one-to-many Paillier. That is 10$\times$
faster than HEROIC (but
branches are farmed out for decision to the remote user across the Internet and
are not counted in that figure).

\mypara{Ascend} \cite{fletcher2012} is retrospectively seen as
the first exemplar of the class of processors designed for
encrypted computing.  This 32 (in plaintext) bit AES-based (co)processor
aimed to work like a black box, literally: the processor was
physically inaccessible and could not be interfered with programmatically
until it had finished an execution task.
It accepted encrypted program and data inputs and
produced encrypted data outputs, and internals were unobservable.
To obfuscate addressing to memory, oblivious RAM
(ORAM) \cite{ostrovsky1990} was integrated.  Timing and power statistics
in terms of what signals appeared on the data pins at what time were arranged
to match specifications set beforehand.  The machine code instruction
set was MIPS RISC (i.e., uncomplicated -- RISC stands
for `reduced instruction set computing' and MIPS (an autonym) is the
variety taught in undergraduate courses) beneath the encryption.  The
operator was considered `semi-honest', in contrast to the
potential adversary of this
paper.  In common with the later designs, every instruction takes the same
time and power to execute no matter what the data.  As in reports on
later projects (except \cite{BB18b}) speeds are not discussed except
with respect to the same processor running without encryption, and a
12{-}13.5$\times$ slowdown is noted.  The authors' reticence is likely
because security audiences do not know that different instruction sets,
processor architectures, platform technologies and compilers are
incomparable, and would be nonplussed by a benchmark that equated to a
4\,MHz Pentium, but is really remarkable.  A 433\,MHz Pentium benchmarks
on Dhrystone at 114$\times$ one DEC VAX 11/780 minicomputer from 1977,
and 4\,MHz Pentium would `wax the VAX.'

\mypara{A modified arithmetic} is the principle behind
later processors for encrypted computing  and is known to be sufficient 
to generate encrypted working \cite{BB13a} since 2013.  HEROIC and
CryptoBlaze both embed the Paillier encryption, which is partially
homomorphic, making the modified addition straightforward 
to do in hardware,
with $x^\E \mathop{[{+}]} y^\E = \E[x+y] = x^\E * y^\E \mod m$. That is
multiplication of the encrypted numbers modulo a 2048-bit integer $m$.
HEROIC's 2048-bit ciphertext multiplication takes 4000
cycles of the supporting 200MHz hardware, so one encrypted (16-bit
plaintext) addition takes 1/50,000s, comparable to a 25\,KHz Pentium
(one 32-bit addition every 1/25,000s).

The KPU uses AES, not a homomorphic encryption (it is reported as having
been tried with Paillier, but the arithmetic proved impossible to
pipeline advantageously), and gets its speed by decrypting internally
once in the pipeline at the start of a sequential train of arithmetic
micro-operations, taking one cycle each, and re-encrypting at the end.
The overhead is reported as 10-20 cycles per train, integrated as
10 pipeline stages. That leads to approx.\
50-60\% pipeline occupation under pressure computing to/from
registers and cache only, and corresponding speed compared to a conventional
machine.

\mypara{Branches} are an obstacle to hardware for encrypted computing and
HEROIC implements a (signed) comparison relation $x^\E \mathop{[{\le}]}
y^\E $ via a lookup table for arithmetic sign (positive or
negative) for branch decisions. It acts as an extra
key.  The table reports whether $x^\E$ has $x$ positive or
negative. In order that the table be small enough, HEROIC's encryption is
one-to-one, not one-to-many.
The table still contains $2^{15}$ rows of 2048 bits (256
bytes) each, so it is 8\,MB.  The architecture is a
stack machine, not a von Neumann machine (the conventional architecture
for modern computers).  That has the advantage that local variables
within a software function are accessed by a (plaintext) number
in the machine code: the stack relative offset of
the variable.  In a conventional architecture, the location would be
accessed via a computed memory address, which would be a 2048-bit
encrypted number because the processor does its computation encrypted,
so $2^{2048}$ memory locations would be addressed and only a few of
those could physically exist.

Cryptoblaze passes branching decisions across the internet to
the user for resolution (after decryption). The KPU resolves
branch decisions internally in the pipeline, via the modified
arithmetic.

\mypara{Memory addressing} is done encrypted in
HEROIC -- it connects only
the last 22 bits of an encrypted address to memory address lines,
which proves sufficient in practice to disambiguate the only $2^{16}$ possible
different 2048-bit addresses.  That is 16MB of addressable memory,
consisting of $2^{16}$ locations each 256B (2048 bits)
wide. Those addresses are scattered randomly through 1GB of physical
RAM, but HEROIC's address translation unit (TLB -- translation
lookaside buffer) remaps them to a 
contiguous, physically backed, 16MB subspace. The KPU remaps 128-bit
encrypted addresses to a designated 32-bit region.

\mypara{The TLB} always has to be a special design.
Processors for encrypted computing have an addressing problem -- it is
impossible to provide physical backing for all the encrypted address
range. Those addresses that do occur in a program must each be
remapped individually to a backed region of memory.
But conventional TLB technology remaps addresses 8192 (a `page') at a
time, so prototypes have to innovate.

A common solution is unit granularity in the TLB, plus dynamic remapping.
Each encrypted address is mapped when it is encountered for the
first time to the next free address in physically backed memory.
Releasing defunct mappings in the TLB is one of the problems for these
processors.

\mypara{A special machine code instruction set} is needed.
HEROIC's, comprises just the one form: $x{\leftarrow}x{\mathop{[{-}]}}y$
conditionally followed by a jump to a point elsewhere in the
program if the result is not positive.  That is computationally complete
\cite{conway87fractran}, and the aim is not only to simplify the
hardware but to make all programs look alike.  Unfortunately HEROIC's
one-to-one encryption undoes that, while compilation for the unusual
instruction set is a challenge.  The newer CryptoBlaze processor adapts
a more conventional architecture and instruction set but the latter
likely does not have the security properties (Box~\ref{b:2}) that are
now seen as necessary in encrypted computing in order to resist
{\em chosen instruction attack} (CIA) \cite{Rass2016}, as explained in the
next paragraph.  The KPU adopts one feature of HEROIC's instruction set
(input and output may be shifted by arbitrary amounts by varying the
instruction constants) and modifies the OpenRISC
(\url{http://openrisc.io}) instruction set in line with that in order
to resist CIA, as explained below. 

The right machine code instruction set is pivotal
when the operator may be an adversary,
as conventional instructions are not secure.
The operator may, for example,  observe an
(encrypted) user datum $x^\E$ and put it through the machine's division
instruction to get $1^\E = x^E\mathop{[/]}x^\E$.  Then any desired
encrypted $y$ may be constructed by repeatedly applying the machine's
addition instruction for $y^\E = 1^\E [+] \dots [+] 1^\E$.
By using the instruction set's comparator
instructions (testing $\E[2^{31}]\mathop{[{\le}]}z^\E$,
$\E[2^{30}]\mathop{[{\le}]}z^\E$, \dots) on an
encrypted $z$ and subtracting on branch, $z$ may be obtained
efficiently bitwise. 
That is the chosen instruction attack (CIA) of \cite{Rass2016}.
If there is no division operator in the hardware then there will 
be a library routine for it (the attacker can write
one themself given an encrypted 1). Failing that, they can try every
encrypted value in a program and its trace and in practice one
of those will be $1^\E$ with probability greater than
$1/2^{32}$, and that gets a $1^\E$ at frequency better than blind
guessing.

The right instruction set resists such attacks.  The KPU's instruction
set contains HEROIC's as a subset and is proved to make those attacks
impossible (Theorem~\ref{e:star} below).

\mypara{The compiler}
must be involved too in order to
avoid {\em known plaintext attacks} (KPAs) \cite{Biryukov2011} based
either on the idea that not only do instructions like $x^\E\mathop{[{-}]}x^\E$
predictably favor one value over others (the result there is always
$x^\E\mathop{[{-}]}x^\E{=}\,0^\E$), but human programmers favor
values like 1.  The compiler must even out the statistics.

The compiler must do so even for object code consisting of a single
instruction.  That gives necessary conditions on 
instruction design and execution shown in
Box\,\ref{b:2} \cite{BB17a}. These constraints must be implemented
by the hardware:
instructions must (\ref{p:2a}) execute atomically, or recent attacks
such as Meltdown \cite{Lipp2018meltdown} and Spectre
\cite{Kocher2018spectre} against Intel become feasible (in those,
memory access instructions in a speculatively executed branch when
aborted leave behind a `halfway-done' taint  in the form of cache lines loaded),
must (\ref{p:2b}) work with encrypted values or an adversary could read them,
and (\ref{p:2c}) must be adjustable via embedded encrypted constants to
offset the values beneath the encryption by arbitrary deltas.  The
condition (\ref{p:2d}) is for the security proofs and amounts to
different padding or blinding factors for encrypted program constants
and runtime values.
\begin{frameenv}[tb]
\begin{flushleft}
{\small
\refstepcounter{mybox}
\rm Box \Roman{mybox}:
Machine code axioms. Instructions \dots 
}
\customlabel{b:2}{\Roman{mybox}}
\end{flushleft}
\begin{enumerate}[leftmargin=*,align=left,labelwidth=1em,labelsep=0.5em,itemindent=\parindent,label={(\alph*)},ref={\Roman{mybox}\protect{\alph*}},itemsep=0.25ex,widest=a,labelindent=0pt]
\item
{\em Are a black box from the perspective of the programming
interface, with no intermediate states}.
\label{p:2a}
\item
{\em Take encrypted inputs to encrypted outputs}.
\label{p:2b}
\item
{\em Are adjustable via (encrypted) embedded constants 
to produce any desired offset delta in the (decrypted, plaintext)
inputs and outputs at runtime}.\kern-5pt
\label{p:2c}
\item
{\em Can have no cipherspace collisions between 
encrypted instruction constants and runtime data.}
\label{p:2d}
\end{enumerate}
\end{frameenv}

In this document (\ref{p:2d}) will be further strengthened to:
\begin{enumerate}[label={},ref={\ref{p:2d}$^*$}]
\item
\begin{minipage}[b]{0.85\columnwidth}\em%
No collisions between constants in different instructions
or different positions.
\end{minipage}
\hfill{\hspace{-3em}\rm ({\ref{p:2d}$^*$})}
\label{p:2d*}
\end{enumerate}
`Different instruction' means different opcodes.
Padding beneath the encryption enforces that, and the
processor silently produces nonsense on violation.  The aim is to block
experiments with transplanted program constants.  With
(\ref{p:2d}), moving runtime encrypted data into
instructions or vice versa was already blocked.

\begin{remark}
HEROIC's $x\leftarrow x \mathop{[{-}]} y$ instruction fails
{\rm(\ref{p:2c})} because 
$x\mathop{[{+}]}C\leftarrow (x\mathop{[{+}]}A)\mathop{[{-}]}(y\mathop{[{+}]}B)$
as {\rm(\ref{p:2c})} requires cannot be achieved by varying the constants in
the instruction, as there are none. That is
$x\leftarrow x\mathop{[{-}]}y \mathop{[+]} K$
where $K= A \mathop{[{-}]}B\mathop{[{-}]}C $ so that 
would be `fixed' for {\rm(\ref{p:2c})} if the instruction included an
extra additive constant $K$. But the subsequent
test $x \mathop{[\le]} 0$ also needs to be fixed to $x \mathop{[\le]}
A$ for $A$ supplied in the instruction. With that,
HEROIC's instruction set would work for the argument and theorem
of {\rm\cite{BB18c}} (below). 
\label{r:1}
\end{remark}
\noindent
The effect of (\ref{p:2a}-\ref{p:2d}) is proved (Appendix, \cite{BB18c}) to be:
\begin{theorem}
A program and its runtime trace  may consistently be interpreted
arbitrarily in terms of data beneath the encryption at
any one point in memory or trace.\kern-2pt
\label{e:star}
\end{theorem}
\noindent
The technical argument shows that picking any point in the trace, so far
as the adversary not privy to the encryption can tell, the word beneath
the encryption may vary over a 32-bit range across recompilations, 
equiprobably.

`Chaotic' compilation always threatens the adversary that a delta
offset has been introduced into runtime data beneath the encryption by
varying the constants in the instruction before and after a point of
interest, because (\ref{p:2a}) and (\ref{p:2b}) prevent the adversary
knowing and (\ref{p:2c}) allows the variation 
(note (\ref{p:2a}) means `no side-channels').

\begin{theorem}[$\rho$CSS]
Relative cryptographic semantic security holds for any one word of data beneath the
encryption and an adversary not privy to the encryption.
\label{e:ddagger}
\end{theorem}
\noindent
That is what is usually rendered as {\em encrypted computation does not
compromise encryption}, but it is really trivial.  If one imagines
the program that does nothing, consisting of no instructions, which
transmits input to output unchanged, all it says is that the input can
be any value (and the output will be the same any value).  There is no
reason or way for an adversary to discern any tendency beneath the
encryption to some proper subset of values.

But data words in a program of any size and form are
individually unconstrained by the adversary's observations (or
experiments, as a continuation of the argument deduces) according to
Theorem~\ref{e:star}.
Further, intuitively the adversary can select any two points in the
program, except a pair as remarked, and they
can be varied independently via changes in the surrounding
instructions that the adversary cannot perceive because
the instruction constants are encrypted, and this paper will quantify
that intuition.

A `chaotic' compiler backs the threat by really
varying runtime data beneath the encryption independently and
arbitrarily across recompilations to the extent the laws of
programming allow, as (\ref{e:S}) ideates.  How the compiler organizes
that is encapsulated in Box~\ref{tb:how}: a new {\em obfuscation scheme}
is generated at each recompilation.
\begin{definition}[Obfuscation scheme]
An {\em obfuscation scheme} is a plan
that specifies {\em a delta from nominal for the data beneath the
encryption in every memory and register location per point in the
program control graph, before and after every instruction}.
\label{d:3}
\end{definition}
\noindent
A high-level, declarative, description of how a compiler works in this
setting is that 
the compiler $\C{-}$ translates, for example, a source code expression
$e$ of type Expr, the value of which is to end up in register $r$ at runtime,
to machine code $\it mc$ of type MC,
and also generates a 32-bit offset $\Delta e$ of (integer) type Off
for $r$ at that point in the
program:
\begin{align}
\C{-}^r &:: {\rm Expr} \to {\rm MC}\times {\rm Off}\notag\\
\C{e}^r &= ({\it mc},\Delta e)
\label{e:1}
\\
\intertext{
Let $s(r)$ be the content of register $r$ in state $s$ of the processor
at runtime. The machine code {\em mc}'s action  is to
change state $s_0$ to an $s_1$ with a ciphertext in $r$ whose plaintext
value differs by $\Delta e$ from the nominal value $s_0(e)$ (the arrow
symbol means `steps eventually to'):
}
s_0 \mathop\rightsquigarrow\limits^{\it mc}  s_1 &~\text{where}~ s_1(r) = 
\E[s_0(e) + \Delta e]
\label{e:2}
\end{align}

\begin{remark}
Bitwise exclusive-or or the binary operation of another
mathematical group are alternatives for `$+$'.
\label{r:group}
\end{remark}

\noindent
For comparison, an `ordinary', non-chaotic, compiler and ordinary
execution platform would instead have the following abstract description:
\begin{align*}
\C{e}^r &= {\it mc}\\
s_0 \mathop\rightsquigarrow\limits^{\it mc}  s_1 &~\text{where}~ s_1(r) = s_0(e)
\end{align*}

\noindent
The `nominal value' $s_0(e)$ is formalized via a canonical
construction. For the encrypted computing context map
variable $x$ to its register $r_x$ (the
runtime value is offset by a delta $\Delta x$), checking the
(ciphertext)  content of $r_x$ in the state and discounting the delta
from the plaintext value to get
$s_0(x) =\D[s_0(r_x)]\mathop{-} \Delta x$.  Arithmetic 
in the expression is formalized recursively, with
$s_0(e_1+e_2)=s_0(e_1)\mathop{+}s_0(e_2)$, etc. In the `ordinary' context,
not encrypted  computing, the nominal value of the variable is instead
$s_0(x) =s_0(r_x)$ with no offset from the value in the
register at runtime, and no encryption in the latter.

The encryption $\E$  is shared with the user and the processor but
not the potential adversaries, the operator and operating system.  The
obfuscation scheme is known to the user, but not the processor and not
the operator and operating system.  The user compiles the program
according to the scheme and sends it to the remote processor with the
encrypted data to execute it on and needs to and does know the offsets
at least on inputs and outputs.  That allows the right data to be
created and sent off for processing and allows sense to be made by the
user of output received, once they have decrypted it.

\begin{frameenv}[t]
\begin{flushleft}
\small
\refstepcounter{mybox}
\rm Box \Roman{mybox}:
\footnotesize
What the compiler does, in sequence:
\customlabel{tb:how}{\Roman{mybox}}
\end{flushleft}
\noindent
\begin{enumerate}[leftmargin=*,align=left,labelwidth=2em,itemindent=\parindent,label={\roman*.},ref={\arabic{mybox}.\protect{\roman*}},noitemsep]
\item
{\sl Generate an obfuscation scheme
of planned data offsets from nominal beneath the encryption. }
\item
{\sl Vary instruction constants to implement (i),
thereby leaving runtime traces unchanged in form, but not content.}
\item
{\sl Equiprobably\,generate\,all\,variations\,(ii),\,hence\,schemes\,(i).}
\end{enumerate}
\end{frameenv}

\section{{Instruction Set}}
\label{s:FxA}

\noindent
As noted in Section~\ref{s:Back}, conventional instruction sets are
not safe against chosen instruction attacks (CIAs) in an encrypted
computing setting.  Without knowing any encrypted constants, it is still
possible to program calculations that give a known constant as answer, such as
$x^\E\mathop{[{-}]}x^\E$, or are biased stochastically towards a known
subset.  But instruction sets satisfying (\ref{p:2a}-\ref{p:2d}) do not
have that problem, by Theorem~\ref{e:star}, so what is needed is a
practical instruction set architecture (ISA) conforming to
(\ref{p:2a}-\ref{p:2d}).  HEROIC's `one instruction' instruction set can
be modified to conform by the incorporation of a couple of encrypted
constants in each instruction, as remarked in Remark~\ref{r:1},
but it is untried and
impractical as a compilation target.

A `fused anything and add' general purpose ISA suitable for encrypted
computing and satisfying conditions (\ref{p:2a}-\ref{p:2d}) is put
forward in \cite{BB17a} as a modification to OpenRISC v1.1
\url{http://openrisc.io/or1k.html}.  A subset is shown in
Table~\ref{tb:1} and in all there are about 200 instructions,
comprising single
and double precision integer and floating point and vector subsets,
uniformly 32 bits long.
Instructions reference up to three of 32 general purpose registers
(GPRs).  There are just two instructions (load/store: {\bf lw}/{\bf sw})
for memory access.  The instruction opcode is in the clear so the
decoding unit in the processor pipeline can act on it, but that allows
an adversary to see what kind of instruction it is, distinguishing addition from
multiplication, etc.  HEROIC's instruction set can
be mimicked by emitting only instruction pairs {\bf
add}\,r$_0$\,r$_0$\,r$_1$\,$k^\E$; {\bf blei}\,i\,r$_0$\,$a^\E$ (the
latter is the one-operand, one-constant form of the {\bf ble} instruction).

\begin{table}[tbp]
\caption{Integer subset of a machine code instruction set for
encrypted working.}
\label{tb:1}
\smallskip
{\centering
\scriptsize
\begin{tabular}[b]{@{}l@{\,}c@{\kern1pt}c@{\kern1pt}l@{\kern1pt}l@{~}l@{~}}
\em op.&
    \multicolumn{4}{@{}l}{\em fields} 
     & \multicolumn{1}{@{}l}{\em \hbox to 0.39in {mnem.\hfill} semantics \hfill}\\
\hline\\[-1ex]
add &$\rm r_0$
  &$\rm r_1$
    &$\rm r_2$
      &\kern0pt$k^\E$
        &\hbox to 0.39in {add} $r_0{\leftarrow}
           r_1\kern0.5pt\mathop{[+]} r_2\mathop{[+]} k^\E$\\
addi &$\rm r_0$
  &$\rm r_1$
      &\kern0pt$k^\E$
        &
        &\hbox to 0.39in {add\:imm.} $r_0{\leftarrow}
           r_1\kern0.5pt\mathop{[+]} k^\E$\\
sub &$\rm r_0$
  &$\rm r_1$
    &$\rm r_2$
      &\kern0pt$k^\E$
        &\hbox to 0.39in {subtract} $r_0{\leftarrow}
           r_1\kern0.5pt\mathop{[-]} r_2\mathop{[+]} k^\E$\\
mul &$\rm r_0$
  &$\rm r_1$
    &$\rm r_2$
      &\kern0pt$k^\E_0\kern0pt k^\E_1\kern0pt k^\E_2$
       &\hbox to 0.39in {multiply}
       $r_0{\leftarrow}(r_1\kern0pt\mathop{[-]}k^\E_1\kern0pt)\mathop{[\,*\,]}(r_2\kern0pt\mathop{[-]}k^\E_2\kern0pt)\mathop{[+]}k^\E_0$\\[0.5ex]
div &$\rm r_0$
  &$\rm r_1$
    &$\rm r_2$
      &\kern0pt$k^\E_0\kern0pt k^\E_1\kern0pt k^\E_2$
       &\hbox to 0.39in {divide} 
       $r_0{\leftarrow}(r_1\kern0pt\mathop{[-]}k^\E_1\kern0pt)\mathop{[\div]}(r_2\kern0pt\mathop{[-]}k^\E_2\kern0pt)\mathop{[+]}k^\E_0$\\
\dots\\
mov&$\rm r_0$&$\rm r_1$& &   
     &\hbox to 0.39in {move}
       $r_0{\leftarrow} r_1$\\
beq& $i$&$\rm r_1$&$\rm r_2$&$k^\E$  
     &\hbox to 0.39in {branch}
       ${\rm if}\,$b$\,{\rm then}\,{\it pc}{\leftarrow}{\it pc}{+}i$,
       $b\Leftrightarrow r_1 \mathop{[=]} r_2 \mathop{[+]} k^\E$\\
bne& $i$&$\rm r_1$&$\rm r_2$&$k^\E$  
     &\hbox to 0.39in {branch}
       ${\rm if}\,$b$\,{\rm then}\,{\it pc}{\leftarrow}{\it pc}{+}i$,
       $b\Leftrightarrow r_1 \mathop{[\ne]} r_2 \mathop{[+]} k^\E$\\
blt& $i$&$\rm r_1$&$\rm r_2$&$k^\E$   
     &\hbox to 0.39in {branch}
       ${\rm if}\,$b$\,{\rm then}\,{\it pc}{\leftarrow}{\it pc}{+}i$,
       $b\Leftrightarrow r_1 \mathop{[<]} r_2 \mathop{[+]} k^\E$\\
bgt& $i$&$\rm r_1$&$\rm r_2$&$k^\E$   
     &\hbox to 0.39in {branch}
       ${\rm if}\,$b$\,{\rm then}\,{\it pc}{\leftarrow}{\it pc}{+}i$,
       $b\Leftrightarrow r_1 \mathop{[>]} r_2 \mathop{[+]} k^\E$\\
ble& $i$&$\rm r_1$&$\rm r_2$&$k^\E$   
     &\hbox to 0.39in {branch}
       ${\rm if}\,$b$\,{\rm then}\,{\it pc}{\leftarrow}{\it pc}{+}i$,
       $b\Leftrightarrow r_1 \mathop{[\le]} r_2 \mathop{[+]} k^\E$\\
bge& $i$&$\rm r_1$&$\rm r_2$&$k^\E$   
     &\hbox to 0.39in {branch}
       ${\rm if}\,$b$\,{\rm then}\,{\it pc}{\leftarrow}{\it pc}{+}i$,
       $b\Leftrightarrow r_1 \mathop{[\ge]} r_2 \mathop{[+]} k^\E$\\
\dots \\
b   &$i$&   &   &   
     &\hbox to 0.39in {branch}
       ${\it pc}\leftarrow {\it pc}+i$\\
sw   &\multicolumn{4}{@{}l@{}}{$(k_0^\E){\rm r_0}~{\rm r_1}$}
     &\hbox to 0.39in {store}
       $\mbox{\rm mem}\llbracket r_0\mathop{[+]}k_0^\E\rrbracket \leftarrow r_1$ \\
lw   &\multicolumn{4}{@{}l@{}}{${\rm r_0}~(k_1^\E){\rm r_1}$}
     &\hbox to 0.39in {load}
       $r_0 \leftarrow \mbox{\rm mem}\llbracket r_1\mathop{[+]}k_1^\E\rrbracket
              $\\
jr   &$\rm r$&  &    &
     &\hbox to 0.39in {jump}
       ${\it pc} \leftarrow r$\\
jal  &$j$&      &    &
     &\hbox to 0.39in {jump}
       ${\it ra} \leftarrow {\it pc}+4;~{\it pc} \leftarrow j$\\
j   &$j$&      &    &
     &\hbox to 0.39in {jump}
       ${\it pc} \leftarrow j$\\[2ex]
\end{tabular}
\qquad
\begin{minipage}[b]{0.3\textwidth}
{\scriptsize\sc Legend}\\[0.5ex]
\scriptsize
\noindent
\begin{tabular}[b]{@{}l@{~}c@{~}l@{~}}
r &--&register indices\\
$k$&--&32-bit integers\\
pc &--&prog.\ count\ reg.\\
$j$&--&program count\\
`$\leftarrow$'&--&assignment\\
ra &--&return addr.\ reg.\\
$\enc[\,~\,]$&--&encryption\\
$i$&--&pc increment\\
$r$ &--& register content \kern-5pt \\
$k^\E$&--&encrypted value $\enc[k]$\\
\multicolumn{3}{@{}l@{}}{$x^\E\mathop{[o]}y^\E =
      \E[x\mathop{o}y]$}\kern-20pt\\
\multicolumn{3}{@{}l@{}}{$x^\E\mathop{[R]}y^\E \Leftrightarrow
     x\mathop{R}y$}\kern-20pt
\end{tabular}
\end{minipage}
}
\end{table}

To make this information concrete for the reader, a runtime trace for
the Ackermann
function\footnote{Ackermann C code: {\bf int} A({\bf int} m,{\bf int} n)
\{ {\bf if} (m == 0) {\bf return} n+1; {\bf if} (n == 0) {\bf return}
A(m-1, 1); {\bf return} A(m-1, A(m, n-1)); \}.} \cite{Sundblad71}
compiled for this instruction set is shown in Table~\ref{tab:3}.  The
machine code is shown disassembled at left, register updates at right.
Encrypted constants are shown with plaintext exposed and padding 
hidden.  The constants in the instructions have randomized plaintexts,
not 0s, 1s, 2s, etc.\ as would be expected.  That goes for the updates
too, except that for readability the delta in the obfuscation scheme for
the return value in register {\bf v0} is set to zero, and the
(encrypted) `13' result can be seen.  Ackermann's is the most
computationally complex function possible, stepping up in complexity for
each increment of the first argument, so getting the answer right is a
confirmation of the correctness of the `chaotic' compilation
technique.  It is short, but the code tests conditionals, assignments,
arithmetic, comparators, call and return.

The 32-bit word-sized instructions may need to embed 128-bit or longer
encrypted constants, so `prefix' words are added as needed, carrying 29
extra bits of data each.

\subsection{Instruction Diddling}

\noindent
Condition (\ref{p:2b}) of Box~\ref{b:2} requires one more constant
in each branch instruction, an encrypted bit $k_0$
that decides if the 1-bit result of the test should be inverted.  The
test is observable by whether the branch is taken or not, so by
(\ref{p:2c}) it should be modifiable by the compiler via an encrypted
instruction constant.  The extra bit changes equals to not-equals and
vice versa, a less-than into a greater-than-or-equal-to, and so on.  The
bit {\em diddle}s the instruction.  In practice, the bit is composed
from the padding bits in the other constants in the instruction, so it
is not explicit in Table~\ref{tb:1}, where the branch semantics shown
are {\em post-diddle}, but the compiler knows what it is.

There is an argument that whether the first program code block after the
branch instruction is the test `fail' or `succeed' case is already
hidden by the general method of `chaotic' compilation applied to boolean
expressions.  That argument is pursued below.

\subsection{The Contestable Equals}
\label{ss:Debate}

\noindent
Diddling works well to disguise less-than instructions and other order
inequalities, but not well for equals versus not-equals. What the
instruction is, equals or not-equals, may be tested by what
proportion of operands cause a jump at runtime.  If almost all do then
that is a not-equals.  If few do then that is an equals.  Trying the
same operand both sides is almost guaranteed to cause equals to fail
because of the embedded constants $k_1$, $k_2$ in \eqref{e:dagger=f}, so
if it succeeds instead, that (seeming!) equality test is (likely) diddled
to not-equals.

So if the test succeeds or not at runtime is detectable in practice
for an equals/not-equals branch instruction, contradicting (\ref{p:2b}).  To
beat that, a compiler must randomly change the way it interprets the
original boolean source code expression at every level so it cannot be
told if the source code, not the object code, had an equality or an
not-equals test.  It internally `tosses a coin' as it works upwards
through a boolean expression for if the source code at that point is to
be interpreted by a truthteller or a liar.  It equiprobably generates,
at each level in the boolean expression, liar code and uses the
branch-if-not-equal machine code instruction for an equality test in the
source code, or truthteller code and uses the branch-if-equal
instruction.  The technique is a generalization of Yao's {\em garbled
circuits} \cite{yao86}, but the compiler works with deeply structured
and recursive logic as well as finite, flat, boolean normal forms of
hardware logic gates.

\def\ME[#1]{\(\hbox{\tt#1}\sp{\E}\)}
\begin{table}[!tbp]
\caption{Trace for Ackermann(3,1), result 13.}
\label{tab:3}
\smallskip
\centering
\begin{minipage}[b]{0.6\textwidth}
\scriptsize
\begin{alltt}
{\rm{PC}}  {\rm{instruction}}                           {\rm{update trace}}
\dots
35  addi t0 a0      \ME[-86921031]      t0 = \ME[-86921028]
36  sub  t1 t1  t1  \ME[-327157853]     t1 = \ME[-327157853]
37  beq  t0 t1  2   \ME[240236822]                  
38  sub  t0 t0  t0  \ME[-1242455113]    t0 = \ME[-1242455113]
39  b 1                                             
41  sub  t1 t1  t1  \ME[-1902505258]    t1 = \ME[-1902505258]
42  xor  t0 t0  t1  \ME[-1734761313] \ME[1242455113] \ME[1902505258]
                                    t0 = \ME[-17347613130]
43  beq  9  t0  t1  \ME[167743945]                 
53  addi sp sp      \ME[800875856]      sp = \ME[1687471183] 
54  addi t0 a1      \ME[-915514235]     t0 = \ME[-915514234] 
55  sub  t1 t1  t1  \ME[-1175411995]    t1 = \ME[-1175411995]
56  beq  2  t0  t1  \ME[259897760]                   
57  sub  t0 t0  t0  \ME[11161509]       t0 = \ME[11161509]   
\dots
143 addi v0 t0      \ME[42611675]      \:\fbox{v0 = \ME[13]}
\dots
147 jr   ra                          \# {\rm{return \ME[13] in{\tt v0}}}
\end{alltt}
\end{minipage}
\qquad
\begin{minipage}[b]{0.3\textwidth}
{\scriptsize\sc Legend}\\[0.5ex]
\scriptsize
\noindent
\begin{tabular}[b]{@{}l@{~}c@{~}l@{~}}
\multicolumn{3}{@{}l@{}}{(registers)}\\
a0 &--& function argument\\
sp &--& stack pointer\\
t0, t1 &--& temporaries\\
v0 &--& return value (framed).
\end{tabular}
\end{minipage}
\end{table}

With that strategy, if the equals branch instruction jumps or not at
runtime does not relate statistically to what the source code says.
Condition (\ref{p:2c}) of Box~\ref{b:2} on the output of the
instruction is effectively vacuous with respect to the source code, as
an observer who sees a jump take place does not know if that is the
result of a truthteller's interpretation of an equals test in the source
code and it has come out true at runtime, or it is the result of the
liar's interpretation and it has come out false. Still, it
might be preferable to remove equals/not-equals.

For other comparison tests, just as many operand pairs cause a branch
one way as the other,\footnote{In 2s complement arithmetic $x<y$ is the
same as $x-y=z$ and $z<0$ and exactly half of the range satisfies
$z<0$, half satisfies $z\le0$.} which makes it indistinguishable
as to whether the opcode is diddled or not.  An equality test cannot be
recreated by an adversary as $x{\le}y$ and $y{\le}x$ because only
$x{\le}y{+}k$ is available, for an unknown constant $k$.  Reversed
operands is allowed by (\ref{p:2d*}) but produces $y{\le}x{+}k$, not
$y{+}k{\le}x$.  An estimate for $k$ might ensue from the proportion of
pairs $(x,y)$ that satisfy the conjunction of the inequality and its
reverse, and in particular $k{<}0$ would be signaled
by the complete absence of pairs that simultaneously satisfy both.
But diddling means the conjunctions might be $x{>}y{+}k$
and $y{>}x+{k}$ instead, and those have no solutions when ${-}k{-}1$ is
negative, not when $k$ is negative.  So equiprobably $k{<}0$ or
$k{\ge}0$, which gives nothing away.

\begin{remark}
A boolean `liar' adds a delta equal to 1 mod 2 to 1-bit data
beneath the encryption, a `truthteller' adds 0 mod 2.
\end{remark}

\vspace{-1.5ex}

\section{{Chaotic Compilation}}
\label{s:Comp}

\noindent
This section will describe in more detail but still declaratively and
abstractly what a chaotic compiler for encrypted computing does,
hopefully pointing out for a security audience just what is difficult
and what is easy about it.

The point of note is that the compiler works with a `deltas' database $D :
{\rm Loc}\,{\to}\,{\rm Off}$ containing an integer offset $\Delta l$ of
integer type Off for data in register or memory location $l$ (type Loc).
The offset is the delta by which the runtime data plaintext beneath the
encryption in the location is to vary from nominal at runtime,
following the description in
(\ref{e:1},\ref{e:2}), and the database $D$ is the incarnation of the
{\em obfuscation scheme} of Defn.~\ref{d:3} for this point in the
program code/control graph.  The compiler has to remember the offset
deltas as it works through the code, and this database serves as
scratchpad.

Routinely, the compiler (any compiler) also maintains a `location'
database $L:{\rm Var}\to {\rm Loc}$ mapping source variables to register
and memory locations.  An intermediate layer in the compiler handles
that and that matter is entirely elided here.

The reader uninterested in the detail that is provided can skip it. But
the detail is required in order to prove what will be claimed, namely
that the compilation method implements the tactic \eqref{e:maxH}.
There is no way of doing that other than by giving detail.

\subsection{Expressions}
\label{ss:Expr}

\noindent
Filling in \eqref{e:1} in more detail, compiling an expression $e$ to
code ${\it mc}_e$ that will get
the result in register $r$ at runtime means the compiler does a
computation
\begin{align}
({\it mc}_e,\Delta r) &= \Ss{D : e}^r
\label{e:3}
\end{align}
where {\it mc} is machine code (type MC), a sequence of machine code
instructions, and $\Delta r$ is the integer offset (type Off) from
nominal beneath the encryption that the compiler intends for the result
in $r$ at runtime.  Recall that $L$ is the location database mapping
source code variables to register and memory locations, and $D$ is
the database containing the `obfuscation vector' at this point in the
code, a list of planned delta offsets at runtime beneath the encryption
per location.
The question is whether the compiler has freedom of choice in choosing
$\Delta r$.  It might be that instructions are not available in the
instruction set by which it could vary $\Delta r$ arbitrarily and
equiprobably across recompilations.

To translate $x{+}y$, where $x$ and $y$ are signed integer
expressions, the compiler 
emits machine code $\hbox{\em mc}_1$
computing expression $x$ in register $r_1$ with offset $\Delta r_1$,
and emits machine code $\hbox{\em mc}_2$
computing expression $y$ in register $r_2$ with offset $\Delta r_2$. By
induction that is:
\begin{align}
(\hbox{\em mc}_1,\Delta r_1) &= \Ss{D : x}^{r_1}
\tag{\ref{e:3}$^x$}
\\
(\hbox{\em mc}_2,\Delta r_2) &= \Ss{D : y}^{r_2}
\tag{\ref{e:3}$^y$}
\end{align}
It decides a random offset $\Delta r$ in $r$ for the whole expression
$e$, emitting the compound code
\begin{align}
\hbox{\em mc}_e &= \hbox{\em mc}_1;\,\hbox{\em mc}_2;
                \,{\bf add}~r ~r_1 ~r_2 ~k^\E\notag
\end{align}
where ${\bf add}\,r \,r_1 \,r_2 \,k^\E$ is the 
integer addition instruction, with 
semantics $r {\leftarrow} r_1 {[+]} r_2 {[+]} k^\E$,
and $k  = \Delta r - \Delta x - \Delta y$ has been designed
so the sum is returned in $r$ with offset $\Delta r$ beneath the
encryption.  That is:
\begin{align}
(\hbox{\em mc}_{x{+}y},\Delta r)
&=
\Ss{D: {x{+}y}}^r 
\tag{\ref{e:3}$^+$}
\end{align}
implementing \eqref{e:3}.
The offset $\Delta r$ is freely chosen.
This construct introduces one `arithmetic instruction that
writes', the {\bf add}, and one arbitrarily mutable offset for it,
$\Delta r$.
That implements the tactic \eqref{e:maxH}. There is nothing special
about the `+' here that has been taken as an example, so \eqref{e:3}
holds inductively of all pure expressions.
The (trivial) base case is
for a simple variable reference $x$ as expression $e$, which 
takes a single trivial `+0' addition instruction to bring it out of the
register $r_x$ to which it is mapped by $L$ and into $r$. The
compiler may substitute an arbitrarily chosen `+$k$' instead of `+0',
thus setting the offset $\Delta r$ in $r$ as it wills, satisfying
\eqref{e:3}.

The construction also gives \eqref{e:maxH}, as $\Delta r$
for $x+y$ is free and exactly one new `instruction that writes' to $r$
(the ${\bf add}\,r \,r_1 \,r_2 \,k^\E$) is involved for it in ${\it
mc}_e$ and the constant in that is freely varied without restriction,
which is sufficient to vary $\Delta r$ freely by \eqref{p:2c}.
By induction every `instruction that writes' in ${\it mc}_x$ and ${\it
mc}_y$ already is freely varied without restriction.  The base case for
a single instruction reference likewise involves one addition
instruction with a freely variable constant.  In conclusion,
\eqref{e:maxH} holds inductively of all pure expressions.

Literal constants (0-ary arithmetic operations) as
expressions are implemented by a completely random choice of value by
the compiler in register $r$.  The database $D$ merely records the
offset from the nominal (`intended') value.

\subsection{Statements}
\label{ss:Stat}

\noindent
The compiler for statements $s$ changes
the database $D$ of deltas at multiple locations. The abstract,
high-level description of what it does in delivering code ${\it mc_s}$
is:
\begin{align}
D' : {\it mc}_s &= \Ss{D : s}
\label{e:4}
\end{align}
A less formally complete exposition will be given than for
expressions, to relieve the reader. It merely has to confirm that
\eqref{e:maxH} is satisfied.
Consider an assignment statement $z{=}e$ (which statement will be called $s$)
to a source code variable $z$, which the
location database $L$ binds in register $r{=}Lz$.
By induction code $\hbox{\it
mc}_e$ for evaluating expression $e$ in temporary register $t_0$ at
runtime is emitted via the expression compiler as in \eqref{e:3} with
$t_0$ for $r$:
\begin{align}
(\hbox{\em mc}_e,\Delta t_0) = \Ss{D : e}^{t_0} \tag{\ref{e:3}$^e$}
\end{align}
A short form add instruction with semantics
$r \leftarrow t_0 \mathop{[+]} k^\E$
is emitted to change offset $\Delta {t_0}$ to a new randomly chosen offset $\Delta' r$
in register $r$:
\begin{align}
{\it mc}_s &= \hbox{\em mc}_e ;\,{\bf add}~r~{t_0}~k^\E \notag
\end{align}
where $k = \Delta' r - \Delta t_0$. That implements
\eqref{e:4} for assignment:
\begin{align}
D': {\it mc}_{z{=}e} &= \Ss{D: z{=}e}
\tag{\ref{e:4}$^{\rm ass}$}
\end{align}
where
the change in the database of offsets is at $r$, to $D' r = \Delta' r$.
The new offset $\Delta' r$ is
freely and randomly chosen by the compiler, supporting \eqref{e:maxH},
and {\em one} new `arithmetic instruction that writes,' the {\bf
add}, is accompanied by {\em
one} new random delta, supporting \eqref{e:maxH} (by induction, 
\eqref{e:maxH} is already true of the code implementing $e$).

The compilation of code constructs if, while, goto, sequence, is
entirely standard and is left as an exercise for the determined reader.
The model of proof above is followed to show \eqref{e:maxH} holds.  A
codicil is that at the end of both branches of conditionals, at the
beginning and end of loops, at source and target of gotos, the offset
deltas in the database $D$ must coincide for reasons of correctness of
the computational semantics, which limits the variability that the
compiler can achieve.  Within that constraint \eqref{e:maxH} is
satisfied.  The information theory is discussed in
Section~\ref{s:Theory}, but the actual code constructions are skipped
here, being clear to `one skilled in the art'.

\subsection{{Ramified Types}}
\label{ss:Long}

\noindent
The problem with types is that there are so many of them, and the
approach to representing them on an encrypted platform is not
intrinsically obvious.  A 64-bit integer could be represented by
encrypting the 64-bit plaintext into a 128-bit ciphertext, for example
(the platform we have used for prototyping is physically 128-bit).  Or
it could be broken into two 32-bit parts that are encrypted separately
as two 128-bit ciphertexts.

Likewise, the way the compiler ought to vary the data is not
intrinsically clear. Should a single 32-bit offset be applied
simultaneously to both 32-bit parts of a  64-bit number? Should
the additive carry be passed between the parts, or ignored?

The answers lie with principle \eqref{e:maxH}: every instruction that
writes in the trace must vary to the fullest extent possible.  With the
instruction set shown in Table~\ref{tb:1}, the same offset when writing
both 32-bit halves of a 64-bit number would mean that the second write
instruction could not vary in a way distinct from the first,
contradicting \eqref{e:maxH}.  So the delta offsets for the two 32-bit
halves of a 64-bit number must be separate and independent.

A similar consideration says that every entry in an array (and every
32-bit part of that entry) must have its own separate, independent,
delta offset.  The problem is that the compiler does not know which
array entry will be accessed at runtime, so cannot in principle
compile for any particular delta offset.  A solution would
be a single, shared 32-bit delta offset for every entry in the array
(and every part of every entry), so the compiler could predict the
offset to apply.  But that runs foul of \eqref{e:maxH}.  It also means
that when one entry is changed, since by \eqref{e:maxH} the write ought
to freely create a new offset, all the other entries in the array ought
to be brought into line for the new array-wide delta offset with a
`write storm', even though the source says they are not being written to.
That might be useful from the point of view of disguising which entry is
intended to be written to, but it is not computationally `efficient' to
have linear time complexity writes to an array.  On the other hand,
reads are constant time complexity, which is attractive (and what a
programmer expects).

The situation is worse again for pointers, which could point anywhere
(the compiler cannot generally predict). The argument
applied above would say that therefore every part of every data
structure must all, universally, share the same delta offset, which
makes nonsense of variability.  The only substantial variation in delta
offsets would be in registers, which pointers cannot reference, and
memory would get a single delta offset to be applied everywhere.
Another approach is needed and  the successful one is discussed below.

\subsubsection{Long types}

Firstly, the reasoning above concluded that to satisfy \eqref{e:maxH},
double length (64-bit) plaintext integers $x$ ought to be treated as
concatenated 32-bit integer `halves' $x=x^H \mathop.  x^L$, the high and
low 32 bits respectively.  The encryption $x^\E$ of $x$ occupies two
registers or two memory locations, containing the encrypted values
$\E[x^H]$, $\E[x^L]$ respectively.  That is not only not obvious but
also needs notation with which to express the corresponding
operational semantics, or the page would fill with $H$ and $L$ superscripts.

\begin{definition}
Encryption of 64-bit integers $x$ comprises encryptions of 
the 32-bit high and low bit halves separately:\kern-5pt
\label{d:2}
\begin{align*}
x^\E = \E[x]= \E[x^H\mathop. x^L]&=\E[x^H]\mathop. \E[x^L]
\end{align*}
\end{definition}
\noindent
Instructions that operate on encrypted 64-bit types contain
(encrypted) 64-bit constants to satisfy (\ref{p:2c}), in order that
64-bit delta offsets across the range can be achieved.
But they may and will be `added' as high+high, low+low separately, with
no carry.  A carry is prohibitively difficult to manage in the compiler
and it is not necessary from the point of view of range, and it is
justified by Remark~\ref{r:group} (any binary operation of a mathematical
group is valid).

\begin{definition}
Let $-^2$ and $+^2$ be the two-by-two independent application of
respectively 32-bit addition and 32-bit subtraction to the pairs of 32-bit
plaintext integer high-bit and low-bit components of 64-bit integers, with
similar notation for other operators. E.g.:
\begin{align*}
(u_1 \mathop. l_1)\mathop{+^2\,}(u_2 \mathop. l_2) &=
              (u_1 \mathop{+} u_2) \mathop.\, (l_1 \mathop{+} l_2)
\end{align*}
\label{d:5}
\end{definition}

\begin{definition}
Let $\tilde{*}$, $\tilde{\hbox{\small+}}$ etc.\ denote
multiplication, addition, etc.\ on 64-bit integers
written as two 32-bit integers.
\end{definition}
\noindent 
Then a suitable atomic encrypted multiplication operation  for the
instruction set working on encrypted 64-bit
operands $x^\E$, $y^\E$ and satisfying (\ref{p:2a}-\ref{p:2d}) gives
the result:
\begin{equation}
\E[(x \,\mathop{-^2}\, k_1) \,\tilde{*}\,
(y \,\mathop{-^2}\, k_2) \,\mathop{+^2}\, k_0]
\tag{$\tilde{*}$}
\end{equation}
Here $k_0$, $k_1$, $k_2$ are 64-bit plaintext integer constants 
embedded encrypted (per Defn.~\ref{d:2}) in the instruction as
$k^\E_i$, $i=0,1,2$.
Putting it in terms of the effect on
register contents, a suitable encrypted 64-bit multiplication
instruction is:
\[
r_0^H\mathop{.}r_0^L \leftarrow (r_1^H\mathop{.}r_1^L \mathop{[-^2]}k_1^\E)
                    \mathop{[\,\tilde{*}\,]}
                  (r_2^H\mathop{.}r_2^L\mathop{[-^2]}k_2^\E) \mathop{[+^2]}
                    k_0^\E
\]
For 64-bit operations the processor
partitions the register set into pairs
referenced by a single name. In those terms, the multiplication
instruction semantics simplifies to:
\[
r_0 \leftarrow (r_1 \mathop{[-^2]}k_1^\E)
                    \mathop{[\,\tilde{*}\,]}
                  (r_2\mathop{[-^2]}k_2^\E) \mathop{[+^2]}
                    k_0^\E
\]
In the instruction set, that is the primitive, atomic instruction
\begin{align*}
{\bf mull}\,r_0\,r_1\,r_2\,k_0^\E\,k_1^\E\,k_2^\E
\end{align*}
following the general assembly format and nomenclature of Table~\ref{tb:1}.
The {\bf l} suffix means it works on `long', i.e., 64-bit, integers.
The other arithmetic instructions follow the
same pattern, and compiled code for long integer
expressions and statements on the encrypted computing platform follows
exactly the form for 32-bit but with `{\bf
l}' instructions.  Just as for 32-bit, exactly one new encrypted
(64-bit) `arithmetic instruction that writes' is issued per compiler
construct, and through it, the 64-bit (i.e., $2\times32$-bit) delta offset
in the target may be freely chosen by the compiler, supporting
(\ref{e:S}) and \eqref{e:maxH}.

\subsubsection{Short Types}
\label{ss:ST}

\noindent
Machine code instructions that work arithmetically on `short' (16-bit)
or `char' (8-bit) or `bool' (1-bit) integers are not needed to compile
the C language at least, because short integers are immediately promoted
to 32-bit for arithmetic.  The compiler generates only internal
accounting for such {\em cast}s.

\subsubsection{{Arrays and Pointers}}
\label{ss:Arr}

\noindent
The natural way to bootstrap integers to arrays {\tt A} of
$n$ integers  is to imagine a set of integer
variables {\tt A}$_0$, {\tt A}$_1$,
\dots one for each entry of the array.  That allows the compiler to translate a
lookup {\tt A}[$i$] as if it were code
\[
\tt{\bf int}~{\tt t}=i;~ \tt({\tt t}=0){\bf?}A_{\rm _0}{\bf:}({\tt t}=1){\bf?}A_{\rm _1}{\bf:}\dots
\]
using a temporary variable {\tt t} and the C ternary operator `\_?\_:\_'.
A write {\tt A}[$i$]=$x$ can be translated as if it were
\[
\tt{\bf int}~{\tt t}=i;~ \tt{\bf if} ({\tt t}=0)~A_{\rm _0}={\it x}~{\bf else}~{\bf if} ({\tt
t}=1)~\dots
\]
The entries in the array get individual
offsets from nominal $\Delta{\tt A}_0$, $\Delta{\tt A}_1$, \dots in the
obfuscation scheme maintained by the compiler. Amazingly, that is 
right according to the discussion with respect to what \eqref{e:maxH}
implies at the beginning of this section, yet it is non-obvious.
One reason why it is non-obvious to a compiler expert is that both on read
and write the scheme is linear time in the size $n$ of the array
(it can, however, easily be improved to log complexity) and array access
`should' be constant complexity. It is also apparently going to be
impractical for pointers, where the expressions and statements above
would have to be $2^{32}$ elements long, as where pointers land at
runtime cannot be predicted.

Nevertheless, the extra complexity may be acceptable in this context --
array lookup ought ideally to at least simulate looking at each entry
(or many of them) in the array in order to disguise which is read, so it
should not be dismissed on that basis.  Multi-core machines may even be
able to execute the component elements simultaneously.

The important technical points of the scheme above are that
(a) the in-processor equality test ignores differences between different
encryption aliases of the index and (b) an invariant set of encrypted
addresses $\E[{\tt\&A}_0]$, $\E[{\tt\&A}_1]$, etc.\ are passed to
memory, so lookup is always to the same place even though memory does
not decrypt addresses (c.f. the discussion in Section~\ref{s:Intro}).
(Memory access can be said to be subject to {\em hardware aliasing}
\cite{Barr98} in the encrypted computing context: i.e., different
encryption aliases of an address access different memory locations, and
(a) and (b) combined beat that.)

The scheme works for pointers {\tt p} too, with lookup {\tt*p}
being compiled like this:
\[
\tt({\tt p}=\&{\tt A}_{\rm 0}){\bf?}{\tt A}_{\rm 0}{\bf:}({\tt p}=\&{\tt A}_{\rm 1}){\bf?}{\tt A}_{\rm 1}{\bf:}\dots
\]
The pointer {\tt p} must be known to be in {\tt A}, so we have
modified C to declare pointers along with a global range that they
point into:
\begin{center}
       {\bf restrict} {\tt A} {\bf int} *{\tt p}{\bf;}
\end{center}
There is still a problem with the scheme in general, however, with
respect to the principle \eqref{e:maxH}. That is that
every access to the $i$th entry ${\tt A}[i]$ is via
precisely the same encryption alias $\E[\&{\tt A}_i]$ as address, and
though it beats the hardware aliasing effect, which memory location it
accesses is visible, hence counts in itself as a `write', yet it does not
vary as it might.
The processor has to repeat exactly the following calculation
to get the address right again and again. An improvement can be made,
but first the code involved has to be listed explicitly (the reader can take
the listing in the next few lines for granted and skip).

Say ${\tt A}$ starts at the $n$th stack location, so ${\tt
A}[i]$ is the $n{+}i$th (assuming word sized entries).
The plaintext address is ${\it sp} {+} (n{+}i)$, where ${\it sp}$ is
the address of the base of the stack. Read should normally be via
this pattern of load word instruction:
\[
   {\bf lw} ~r~\E[n{+}i]({\bf sp}).
\]
That causes the processor to sum $(n+i)^\E \mathop{[+]}
s^\E$ where $s^\E$ is the value in {\bf sp}. It passes
the result as effective address.
But the value in $s$ differs from its nominal value $\it sp$ by
a $\Delta {\bf sp}$ planned by the compiler, so the
read instruction must be:
\[
   {\bf lw} ~r~\E[(n{+}i){-}\Delta{\bf sp}]({\bf sp})
\]
The same cipher\-text $\E[(n{+}i){-}\Delta{\bf
sp}]$ must be used as the constant at every read.

For write, the compiler emits the corresponding store:
\[
   {\bf sw}~\E[(n{+}i){-}\Delta{\bf sp}]({\bf sp})~r
\]
That works around the `hardware aliasing' effect in encrypted computing
but does not support
\eqref{e:maxH} because the effective address could 
be varied, as explained below.

\subsubsection{Varying Addresses}
\label{ss:VA}

\noindent
To satisfy \eqref{e:maxH} the compiler should vary the address used at
every write, by choosing a new encryption alias for $\E[\&{\tt A}_i]$ so a
new memory location is written.  Reads will be from there till the next
write.  That does satisfy (\ref{e:S}).

An array or variable on the heap instead of stack
necessitates the heap base address register ({\bf zer}) instead of {\bf
sp} in the load/store instructions, otherwise code is the same.

Unfortunately, using a new address all the time quickly fills the
address mapping cache (the TLB) at runtime with addresses that will
never be used again yet occupy translation slots. The solution is
given below and again involves the compiler.

At the ends of loops (and after the then/else blocks of conditionals, at
return from a function, at the label target of {\bf goto}s, and wherever two
distinct control paths join) the compiler issues 
instructions to restore the original address used by copying the data
back to there from the address currently in use. The original delta
offsets also have to be restored, but we will suppose that is done
separately.
Say the variable in question's location is nominally the $n$th on the
stack. The resynchronization instruction sequence is
\[
{\bf lw} ~r~\kappa_1({\bf sp}); ~{\bf sw}~\kappa_0({\bf sp})~r
\]
where $\kappa_0$, $\kappa_1$ are the encryption aliases for $\E[n{-}\Delta{\bf
sp}]$ in use at the beginning and end of the loop respectively.
That reads from the one address then writes to the other.

The solution to the problem that varying the address used to access
arrays and variables fills the TLB with mappings that will not be used
again is that the compiler issues an instruction sequence to remove the
mapping for the defunct address:
\[
   {\bf addi}~ r~{\bf sp}~ \kappa_1;\,{\bf mtspr}~{\bf UDTLBEIR}~r;
\]
The addition instruction reproduces exactly the processor pipeline
calculation that forms the effective address, and the
`move to special purpose register' ({\bf mtspr}) instruction sends it to
the special `user data TLB entry invalidate register' ({\bf UDTLBEIR}),
which affects the TLB.

The register needs an instruction to be executed in operator mode for
access to succeed, so a system call is required, but there is no
information leak because in user mode the same instruction would be used
and it only carries the effective address, which is visible by which
memory location is accessed.  The compiler can save up these sequences
till the end of a code block or the function body in order to keep 
defunct entries longer in the TLB (the advantage is that of bank robbers
who shake their pursuers by swapping getaway cars in a busy car park
instead of a quiet cul-de-sac).

\subsubsection{{Structs and Unions}}
\label{ss:Struct}

\noindent
C `structs' are records with fixed fields.  They cause no problem for
the compiler at all.  It treats each field in a variable of struct type
like a separate variable.  That is, for a variable {\tt x} of struct
type with fields\ {\tt.a} and\ {\tt.b} the compiler invents variables,
{\tt x.a} and {\tt x.b}.

For an array {\tt A} with $N$ entries that are structs, the
compiler invents $2N$ variables ${\tt A}_i{\tt.a}$ and ${\tt
A}_i{\tt.b}$. Access to ${\tt A}[i]{\tt.b}$ causes the
compiler to emit code that tests only the {\tt.b} 
addresses in the range dedicated to {\tt A}, half of the total.

Unions have a surprise.  The correct code for accessing a union member
is exactly that for accessing a variable {\tt x.b} sited at the start of
the union {\tt x}.  But they also provide another indication that the
correct way to treat arrays and other long types is via one delta offset
per entry, not one delta offset shared for every entry, despite the
compiler problems and inefficiencies it causes. The reason is that a
union of an array with a struct (a common programming meme) would force
all the struct's fields to the same (single, unique) delta offset as the
array. That goes against the principle \eqref{e:maxH}.

\subsection{Memory Example}
\label{ss:Ex}

\newenvironment{code}{\begin{scriptsize}\begin{alltt}}{\end{alltt}\end{scriptsize}}
\def\MY[#1]{\(\E[\hbox{\tt#1}]\)}
\def\la{\(\leftarrow\)}

\begin{table*}[!tp]
\caption{Trace for sieve showing hidden padding bits in data (right).
Stack read and write instruction lines are in red, address base
(register content, right) in violet and address displacement
(instruction constant, left) in blue.}
\label{tab:4}
\flushleft
\begin{tabular}{@{}l@{}r@{}}
\begin{minipage}{0.95\textwidth}
\begin{alltt}\scriptsize
{\rm{PC}}    {\rm{instruction}}                         {\rm{trace updates}}         | {\rm{hidden}}
\dots
{\color{red}19300 addi t1 sp \MY[-407791003]         t1 \la \MY[{\color{violet}-866593752}|1532548040]}  \,#{\rm write local array}
{\color{red}19320 sw   \MY[{{\color{blue}866593746}}](t1) t0         mem[\MY[-6|-712377144]]}           #{\rm }a[7]{\rm at}{\it sp}{\rm+40}
{\color{red}                                         \la \MY[-866593745|1800719299]}
\dots
{\color{red}20884 addi t1 sp \MY[-1763599776]        t1 \la \MY[{\color{violet}2072564771}|-1935092797]} \,#{\rm write local variable}
{\color{red}20904 sw   \MY[{{\color{blue}-2072564772}}](t1) t0       mem[\MY[-1|1518992593]]}           #{\rm }i{\rm at}{\it sp}{\rm+45}
{\color{red}                                         \la \MY[2072564779|-1773201679]}
\dots
22340 addi t1  sp  \MY[-418452205]       t1 \la \MY[-877254954|1532548040]
22360 bne  t0  t1  84
{\color{red}22384 addi t1  sp  \MY[-407791003]       t1 \la \MY[{\color{violet}-866593752}|1532548040]}  #{\rm read local array}
{\color{red}22404 lw   t0  \MY[{{\color{blue}866593746}}](t1)        t0 \la \MY[-866593745|1800719299]}  #{\rm }a[7]{\rm at}{\it sp}{\rm+40}
22424 addi t0  t0  \MY[-1668656853]      t0 \la \MY[1759716698|1081155516]
22444 b    540
22988 addi t1  zer \MY[1759716697]       t1 \la \MY[1759716697|1325372150]
23008 bne  t0  t1  44
\dots
{\color{red}23128 addi t0  sp  \MY[-1763599776]      t0 \la \MY[{\color{violet}2072564771}|-1935092797]} #{\rm read local variable}
{\color{red}23148 lw   t0  \MY[{{\color{blue}-2072564772}}](t0)      t0 \la \MY[2072564779|-1773201679]} #{\rm }i{\rm at}{\it sp}{\rm+45}
23168 addi t0  t0  \MY[1723411350]       t0 \la \MY[-498991167|-981581771]
23188 addi t0  t0  \MY[-1862832992]      t0 \la \MY[1933143137|-1629507929]
23208 addi v0  t0  \MY[-1933143130]     \fbox{v0 \la \(\E[7\)}        \:|1680883739\(]\)  #{\rm return}
\dots
23272 jr   ra
STOP
\end{alltt}
\end{minipage}
&
\end{tabular}
\end{table*}

Running a Sieve of Eratosthenes program\footnote{Sieve C code:
{\bf int}
S({\bf int} n\,) \{ {\bf int} a[N]=\{[0\dots N-1]=1,\};
    {\bf if} (n$>$N$||$n$<$3) {\bf return} 0;
    {\bf for} ({\bf int} i=2; i$<$n; ++i) \{
        {\bf if} (!a[i]) {\bf continue};
        {\bf for} ({\bf int} j= 2*i; j$<$n; ++j) a[j]=0;
        \}; 
    {\bf for} ({\bf int} i=n-1; i$>$2; {-}{-}i)
        if (a[i]) {\bf return} i;
    {\bf return} 0;
\}
.
}
for primes shows up well how memory accesses are affected by encrypted
address aliasing.

The final part of the trace  is shown in Table~\ref{tab:4} with
two reads from elements on the stack shown in red. The
address base (in register) and address displacement (a constant in the
load/store word instruction) are shown in violet.
The assignments to these stack locations are up-trace and do have the
same address base and displacements as in the later reads.
The plaintext addresses -6, -1 reflect the fact that the stack grows down
from top of memory (\mbox{-1}), but it is the $\E[$-6$|$-712377144$]$
(address -6, padding -712377144) and
$\E[$-1$|$1518992593$]$, the encryptions of the compound of address and
hidden bits, that  are passed as the effective addresses to the
memory unit.

\section{Information Theory of Obfuscation}
\label{s:Theory}

\noindent
In this section the degree of independence of data beneath the
encryption in a runtime trace for a program compiled satisfying
\eqref{e:S} and \eqref{e:maxH} will be quantified, improving
on the known $\rho$CSS result in \cite{BB18c}.
A {\em trace} $T$ is the runtime sequence of
writes to registers and memory locations.  If a location is read for the
first time without it having previously been written in the trace, then
that is an {\em input}.  There are no relevant differences in instruction
order or kind (opcode) in this context \ref{p:1b}.

Trace $T$ is a stochastic random variable, varying across
recompilations of the same source code by a
chaotic compiler.  The compiler chooses obfuscation schemes 
as described in previous sections, and the probability distribution
for $T$ depends on the distribution of those choices.  After an
assignment to a register $r$, the trace is longer by one: $T' = T
{}^{\frown} \langle r{=}v^\E \rangle$.
Let $\H({T})$ be the {\em entropy} of {\em trace} $T$ in this
setting.  I.e., let $f_T$ be the probability distribution of $T$, the
entropy is the expectation
\begin{equation}
\H(T)=\mathds{E}[-\log_2 f_{T}]
\end{equation}
The increase in
entropy from $T$ to $T'$ (it cannot decrease as $T$ lengthens) is
informally the
number of bits of unpredictable information added.
Only these fragments of information theory will be required:

\begin{fact}
The flat distribution
$f_{X}{=}1/k$ constant is the one with maximal entropy
$\H({X}){=}\log_2 k$, on a signal $X$ with $k$ values.
\label{f:1}
\end{fact}

\begin{fact}
Adding a maximal entropy signal to any random variable on a $n$-bit
space ($2^n$ values) gives another maximal entropy, i.e., flat, distribution.
\label{f:2}
\end{fact}

\noindent
If the offset $\Delta r$ beneath the encryption is chosen randomly and
independently with flat distribution by the compiler, so it has maximal
entropy, then $\H(T')=\H(T)+32$, because there are 32 bits of
unpredictable information added via the 32-bit delta to the 32-bit
value beneath the encryption, so the 32-bit sum of value plus delta
varies with (32-bit) maximal entropy.

Although per instruction the compiler has free choice in accord with
\eqref{e:maxH}, not all the register/memory write instructions issued by
the compiler are jointly free as to the offset delta for the target
location -- it is constrained to be equal at the beginning and end of
a loop, and in general at any point where two control paths join
(\ref{p:5b}):

\begin{definition}
An instruction emitted by the compiler that adjusts the offset in
location $l$ to a final value common with that in a joining control path
is a {\em trailer} instruction.
\end{definition}
\noindent
Trailer instructions come in {\em sets} for each location $l$ at a
control path join, with one member on each path.  Each is last to write
to $l$ before the join.
In particular, there are trailer instructions before 
return from a subroutine.

Because running through the same instruction or a different instruction
with the same delta offset for the target location a second time does
not add any new entropy (the delta is determined by the first
encounter), the total entropy in a trace can be counted as follows:

\begin{lemma}
The entropy of a trace compiled according to \eqref{e:maxH} is $32(n+m)$
bits, where $n$ is the number of distinct arithmetic instructions that
write in the trace, counted once only per set if they are one of a set
of trailer instructions for the same location, and once each if they are
not, and $m$ is the number of input words.
\label{l:3}
\end{lemma}
\noindent
Recall `input' is provided by those instructions that read
first in the trace from a location not written earlier in it.

Observing data at any point in the trace sees variation across
recompilations.  The principle \eqref{e:maxH} asserts that every
opportunity provided by an arithmetic instruction that writes is taken
by the compiler to introduce new variation.  At `trailers'
the compiler organizes several instructions to synchronize final
deltas in different paths but that is sometimes unnecessary because a
location will be rewritten before it is ever read again.  In such cases,
the variation the compiler introduces is not maximal because it could
be increased by varying deltas independently.
So consider that compiler constructions might be embedded in a context
that reads all locations.  Then trailer synchronization is necessary and
the compiler introduces the maximal entropy possible:

\begin{proposition}
The trace entropy of context-free compiler constructions
that conform to \eqref{e:maxH}
is maximal with respect to
varying the constants in the machine code.
\label{t:3}
\end{proposition}
\noindent
The proposition implies at least 32 bits of entropy in the variation
beneath the encryption must exist in any location $l$ 
where (1) the location has been written, or (2) read without a
prior write.  In (1) the datum is written by an instruction and the
compiler generates variations in the obfuscating delta $D' l$ in the
obfuscation scheme $D'$ after the instruction, or is copied exactly from
somewhere else that the compiler influences in that way. In (2) it is an
input, which is subject to planned variations $D_0 l$ that must be
satisfied by the provider of the input.

The\,following\,is\,obtained\,by\,structural\,induction\,in\kern-2pt\cite{BB17a}:\kern-10pt

\begin{corollary} {\rm(\ref{e:S})}
The probability across different compilations by a compiler that
follows principle \eqref{e:maxH}
that any particular 32-bit value has encryption $\E[x]$ in
a given register or memory location at any given point in the program
at runtime is uniformly $1{/}2^{32}$.
\label{t:4}
\end{corollary}

\noindent
That formally implies Theorem~\ref{e:ddagger}, relative to the
security of the encryption. But a stronger result can now be obtained
from the lemma and proposition above:

\begin{definition}
Two data observations in the trace are (delta) {\em dependent} if
they are of the same register at the same point, input and output
of a copy instruction, or the same register 
after the last write to it in a control path before a join and
before the next write.
\label{d:8}
\end{definition}
\noindent
The variation in the trace observed at two (or $n$) independent
points is maximal possible:
\begin{theorem}
The probability across different compilations by a compiler that
follows principle \eqref{e:maxH} that any $n$ particular
32-bit values in the trace have encryptions $\E[x_i]$, provided they are
pairwise independent, is $1{/}2^{32n}$\kern-2pt.
\label{t:1}
\end{theorem}

\noindent
Each dependent pair reduces the entropy by 32 bits.

\section{Discussion}
\label{s:Discuss}

\noindent
Theorem~\ref{t:1} quantifies exactly the correlation that exists
in data beneath the encryption in a trace where the compiler
follows the principle \eqref{e:maxH} (every arithmetic
instruction that writes is varied to the maximal extent possible across
recompilations).  It names the points in the trace where the compiler's
variations are weak and statistical influences from the original source
code may show through.  For example, if the code runs a loop summing the
same value again and again into an accumulator, then looking at the
accumulator shows an observer $\E[a+ib+\delta]$ for a constant offset
$\delta$ and increasing $i$.  That is an arithmetic series with unknown
starting point $a+\delta$ and constant step $b$ and it is likely to be
one of the relatively few short-stepping paths, with small $b$. That
knowledge
can be leveraged into a stochastically based attack on the encryption.
But if the encryption has no weakness to that vector 
then there is no danger.  Such a characteristic of the encryption
would be expressed as `there is no polynomial time method that
determines $a$ or $b$ from a sequence $\E[a+ib]$ with probability
significantly greater than chance' (as block size $n \to \infty$).

A compiler following the principle \eqref{e:maxH} does as well as any
may to avoid weaknesses based on relations such as the above between
data at different points in the runtime program trace.  The only way to
eliminate them completely is to have no loops or branches in the object
code, by Theorem~\ref{t:1}.  That would be a finite-length calculation
or unrolled bounded loop with branches bundled as $\E[tx+(1-t)y]$ into
the calculation, where $x^\E$ and $y^\E$ are the outcomes from the two
branches executed separately, and $t$ is the boolean test result.
(Those are essentially the calculations available via FHEs.)

That speaks to classical concepts of obfuscation and security
via the following argument, which shows that the operator
does not win game G2 of Section~\ref{s:Key}.

\begin{claim}
Encrypted computing as described in this paper is
resistant to polynomial time (in the block/word size $n$) attacks 
by the operator on the runtime data beneath the encryption, provided
the encryption itself is resistant to such attacks.
\end{claim}

\begin{proof}[Proof sketch]
Suppose the adversarial operator has a polynomial time (in
the number of bits $n$ in a word on the platform) method of working out
what the data beneath the encryption is in register $r$ at some
identified point in the trace of program $P$.  That point of interest
may even move (polynomially) with $n$, being specified perhaps as
`the point of last change in $r$ in the first $n^3$ steps'.  The
operator knows program $P$ and may have suggested it themself,
before the user compiles it.  The operator can see the compiled maximal
entropy code $\C{P}$, and see it running and probe it by running it
at will.

The user readies a sequence of compilations $\C{P_n}$ of $P$, as
described in this paper, with the $n$th being for a $n$-bit platform as
target, and $P$ having been partially or completely unrolled as $P_n$
with no loops or branches in the first $e^n$ (i.e., super-polynomially
many) machine code instructions.  If the program predictably ends before
then, it is to be unrolled completely.  These are compilations of the
same program $P$ all with the same end-to-end semantics that could be
produced entirely automatically by a monolithic compiler incorporating
the unrolling and branch bundling in its front end.  The operator is
invited to apply their method and predict the values beneath the
encryption at the chosen points in the runtime trace(s) of these
compiled codes, which may differ only in consequence of the number $n$
of bits in a word on the platform (i.e., the theory of arithmetic mod
$2^n$).

Theorem~\ref{t:1} implies the operator's method cannot exist as
follows.  There are no loops or branches (hence no delta dependencies,
per the terminology of the theorem and Defn.~\ref{d:8}) in the part of
the trace the observer has time to examine.  The theorem says the
compiler will have arbitrarily and independently varied what is meant by
the values throughout that length of the trace by varying the deltas
from the nominal value independently across each instruction in turn.
So there is no reliable relation between the data and the operator cannot 
use it as a decryption aid.  The operator is reduced to attacking the
encryption with no further information.  The encryption on its own is
hypothesized to be resistant to polynomial time attacks, and the claim
is proved by contradiction.
\end{proof}

The credibility of the argument is supported by the trivial case in
which the program unrolls completely.  Then it is equivalent to a logic
circuit in hardware, but with the data on every pin and wire converted
to encrypted form.  It is known from the theory of Yao's garbled
circuits \cite{yao86} that the intended values cannot be deciphered
without the garbling scheme, which equates to an obfuscation scheme of
deltas here: `+1 mod 2' is boolean negation on a 1-bit boolean
value, while `+0 mod 2' leaves the 1-bit boolean value unchanged.  The
obfuscation scheme is known only to the user, not the operator.

The argument could have been made in 1986 (the year of publication
of \cite{yao86}), if the hardware and electronic engineers had moved
their idea on into general computation.

A subtlety is that encryption appears not to be necessary for the
argument if one substitutes for `data beneath the encryption'
with `value as really intended by the user'.  Yao's garbled
circuits are already garbled without any encryption and equally the code
produced by the chaotic compiler is `obfuscated', with an obfuscation scheme
determining randomly chosen offsets from nominal at every point in the
trace.  It is a matter of conjecture for an observer as to what the
user/programmer really meant when 2 is in the trace, because if the
obfuscation scheme has -1 as delta at that point, then 3 was really
meant by the user, and if -12 is the delta, then 14 was really meant.
All options are equally feasible, and the compiler has rendered them 
equally probable. The practical problem is that the mode of 
compilation -- with exponential unrolling -- in the argument for the
Claim of this section leads to unfeasibly long machine code programs,
and in any case there is no guarantee that any particular word size $n$
is not vulnerable, only that in the limit any particular attack method
fails. But that may be moot too as hardware platforms are not
arbitrarily extensible in word length (perhaps they may become so).
Simulating a platform with word length $n$ on a 32 bit platform
does not satisfy the axiom of atomicity in Box~\ref{b:2}.

As a codicil, the treatment of short integer types here (they are
promoted to standard integer length) prompts the question of whether
entropy could be increased by changing to 64-bit or 128-bit plaintext
words beneath the encryption, instead of 32-bit, and correspondingly
sized delta offsets from nominal.  That logic appears correct.  The
32-bit range of variation of standard-sized integers would be swamped by
a 64-bit delta introduced by the compiler and the looped stepping
example $\E[a + ib + \delta]$ at the beginning of this section would
have a 64-bit $\delta$, so would have $2^{64}$ possible origin points
for the path for any hypothetical step $b$, not $2^{32}$.

\section*{Implementation}

At the current stage of development, our own prototype compiler 
(see \url{http://sf.net/p/obfusc}) has near
total coverage of {\sc ansi} C with GNU extensions, including
statements-as-expressions and expressions-as-statements.  It lacks {\bf
longjmp}, and computed {\bf goto}.
It is being debugged via the venerable {\em gcc\/}
`c-torture' test-suite v2.3.3
(\url{http://ftp.nluug.nl/languages/gcc/old-releases/gcc-2/gcc-2.3.3-testsuite.tar.bz2}),
and we are presently about one quarter way through that.

\section{{Conclusion}}

\noindent
How to compile all of {\sc ansi} C for encrypted computing without
decryption on the memory address (or data) path has been set out here.
That opens up the field for software development using the canonical
toolchain of compiler, assembler, linker, loader with operating system
support.  It allows processors for encrypted computing to access memory
at normal speeds.  The only hardware support needed is a unit
granularity address translation lookaside buffer that remaps encrypted
addresses first-come, first-served to physically backed memory.  The
compiler inserts instructions that release the mappings.

The technical difficulty is that encrypted addresses vary for the same
plaintext address and distinct calculations for the same address
produce a different ciphertext variant.  A systematic coding discipline
is followed that overcomes that.  The downside of the compiler solution
is that array and pointer access are linear or log time in the array
size, not constant time, but the upside is that all modes of access,
e.g., via pointer or displacement in array, are as compatible as the
programmer expects them to be, to any structural depth.

Further, the compiler has to randomly vary the code it generates as much
as possible in order to provide security guarantees.  A single principle
for the compiler to follow has been enunciated -- any arithmetic
instruction that writes must be varied by the compiler to the maximal
extent possible.  It has been shown that then the compiler is `best
possible' in terms of introducing maximal possible entropy across
recompilations to the data beneath the encryption in a runtime trace,
and that swamps biases introduced by human programmers or other
agencies.  The theory quantifies exactly an existing `cryptographic
semantic security relative to the security of the encryption' result for
encrypted computing, and implies that the adversarial operator or
operating system cannot guess what user data is beneath the encryption,
to any degree better than chance.  That is perhaps so even in the
absence of encryption, thanks to an `obfuscation scheme' the compiler
modifies the user's data with, which an adversary has no basis for
distinguishing from the user's intention.

\renewcommand{\baselinestretch}{0.93}
\bibliographystyle{IEEEtran}
\begin{scriptsize}
\bibliography{IEEEabrv,ches-2019}
\end{scriptsize}
\renewcommand{\baselinestretch}{1}

\appendix
\def\appendixname{Appendix -- not for publication}
\small

\section*{Appendix}

This appendix contains proofs from work that cannot be referred to
without breaking the anonymity rules, and/or that it seems better to
avoid interrupting the text with.  Accordingly, the referees should
perhaps regard it as additional material that may be incorporated into a
final text or referenced from there, as they may require.

\setcounter{theorem}{0}
\setcounter{remark}{0}
\setcounter{definition}{0}
\setcounter{objective}{0}

\def\theproposition{\thesection\arabic{theorem}}
\def\thetheorem{\thesection\arabic{theorem}}
\def\thelemma{\thesection\arabic{theorem}}
\def\thecorollary{\thesection\arabic{theorem}}
\def\theremark{\thesection\arabic{remark}}
\def\thedefinition{\thesection\arabic{definition}}
\def\theobjective{\thesection\arabic{objective}}

\section{Proofs}
\label{s:Proofs}

\begin{proposition}
\label{at:1}
No method of observation exists by which the operator (who does not 
possess the key) may decrypt program output from 
the `fixed' HEROIC instruction set of  Remark~\ref{r:1}.
\end{proposition}

\begin{proof}
(Sketch) Suppose program $C$ is written using only the instructions
addition of a constant $y{\leftarrow}x[+]k^\E$ and branches based on comparison
with a constant $x[<]K^\E$ (a `fixed' HEROIC set), which satisfy
(\ref{p:2a}-\ref{p:2d}) of Box~\ref{b:2}.  The hypothetical method takes
as inputs the trace $T$ and code $C$.  But a modified code $C^*$ is
constructed below such that (i) it has a trace $T^*$ that `looks the
same' as $T$ to the operator, modulo encrypted data, and (ii) the new
code $C^*$ `lo\-oks the same' as $C$, modulo embedded encrypted
constants, so the operator's method must give the same result applied to
$C^*$ and $T^*$ as it does applied to $C$ and $T$, which is $y^\E$, say.
But the code $C^*$ gives the output $\E[y{+}7]$ when run, not $y^\E$.
So the method does not work.

The program $C^*$ 
differs from $C$ only in the encrypted constants $K^\E$ in 
the branch instructions.  Otherwise it is the same as $C$.
The constants $K$  need changing to match $x$ taking a
value that is 7 more than before beneath the encryption.  Changing $K$ to
$K^*=K{+}7$ achieves that. Branches jump (or not) as
they did before the increase of the plaintext data everywhere by 7. The
addition instructions are consistent as they 
are with the plaintext increase in both
input and output.  So the code does
the same at runtime as $C$ does, but on data that is everywhere
$x[+]7^\E$ instead of $x$, as required for the contradiction.
\end{proof}

\begin{corollary}
There is no method by which the privileged
operator can alter program $C$ using just add and compare
with constant instructions to get output $y^\E$ for known $y$.
\label{t:2}
\end{corollary}

\begin{proof}
Suppose for contradiction that the operator builds program
$C^*{=}f(C)$ that returns $y^\E$.  Then its constants $k^\E$ and
$K^\E$ are found in $C$, because $f$ has no way of
arithmetically combining them (the no collisions condition (\ref{p:2d}) means
they cannot be combined arithmetically in the processor and the operator
does not have the encryption key).  Proposition~\ref{at:1}
says the operator cannot read $y$ from the output of $C^*$, yet knows what
it is. Done by contradiction.
\end{proof}

\begin{theorem}
There is no method by which the privileged operator can read plaintext
runtime data
from a program $C$ built from instructions satisfying
(\ref{p:2a}-\ref{p:2d}), nor
deliberately alter it to give an intended
output $y^\E$ with $y$ known.
\label{at:3}
\end{theorem}

\begin{proof}
(Sketch)
A modified code $C^*$ is constructed that
looks the same modulo encrypted constants, and has
runtime trace $T^*$ that looks the same as the original
$T$ modulo encrypted data.
The argument goes as for 
Proposition~\ref{at:1} and Corollary~\ref{t:2}.

In program $C$, every arithmetic instruction of the form 
\[
r_0{\leftarrow}(r_1[-]k_1^\E) \mathop{[\Theta]}
(r_2\mathop{[-]}k_2^\E)\mathop{[+]}k_0^\E
\]
for operator $\Theta$
can be chang\-ed for $C^*$ via adjustments in its embedded
constants to accommodate every data value passing
through registers and memory to be +7 more beneath the encryption
than it used to be, as in the proof of
Theorem~\ref{at:1} and Corollary~\ref{t:2}. The change is from 
$k_i$ to $k_i^*{=}k_i+7$, $i=0,1,2$.

A branch instruction in $C$ with test
$(r_1\mathop{[-]}k_1^\E)\mathop{[R]}(r_2\mathop{[-]}k_2^\E)$ for
relation $R$, the instruction is changed for $C^*$ to
$(r_1\mathop{[-]}\E[k_1^*])\mathop{[R]}(r_2\mathop{[-]}\E[k_2^*])$
with $k_i^*=k_i+7$, $i=1,2$, and
the branch goes the same way at runtime in trace $T^*$ for
$C^*$ as it did originally in trace $T$ for $C$.
Unconditional jumps are unaltered.

The outcome is a trace $T^*$ that is the same as $T$
modulo the encrypted data values, which by hypothesis
cannot be read by the adversary (they differ by 7 from the originals,
beneath the encryption).
Code $C^*$ looks the same too, apart from the embedded (encrypted)
constants, which also cannot be read by the adversary. As in
the earlier proof, a method $f(C,T)$ for decryption must give the same result
as $f(C^*,T^*)$, yet the answers (the decrypted data) are different by 7 in the
two cases, so method $f$ cannot exist.
\end{proof}

\begin{lemma}
There is a compile strategy for machine code instruction
sets satisfying (\ref{p:2a}-\ref{p:2d}) such that
the probability across different compilations
that any particular 32-bit value $x$ has its encryption $x^\E$ in
a given register or memory location at any given point in the program
at runtime is uniformly $1{/}2^{32}$\kern-3pt.
\label{at:4}
\end{lemma}

\begin{proof}
Consider the arithmetic instruction $I$ in the program.
Suppose that by modifying the embedded constants
in the other instructions in the program it is already possible for
all other locations $l$ other than that written by $I$ and at all other
points in the program to vary the
value $x_{l} = x{+}\Delta x$, where $x_{l}^\E$ is stored in $l$,
randomly and
uniformly across compilations, taking advantage of the properties of
the instruction set as the compiler described in the text does.
Let $I$ write value $y^\E$ in location $l$.  By design,
$I$ has a parameter $k^\E$
that may be tweaked to offset $y$ from the
nominal result $f(x+\Delta x)$ by any chosen amount $\Delta y$.
The compiler chooses $k$ with a
distribution such that $\Delta y$
is uniformly distributed across the possible range. The
instructions in the program that receive $y^\E$ from $I$ may be
adjusted to compensate for the
$\Delta y$ change by changes in their controlling parameters.
Then $p(y{=}Y){=}p(f(x{+}\Delta x){+}\Delta y{=}Y)$ and the
latter probability is $p(y{=}Y){=}\sum\limits_{\Upsilon}
p(f(x{+}\Delta x){=}\Upsilon \land \Delta y{=}Y{-}\Upsilon)$.
The probabilities are independent (because  $\Delta y$ is newly
introduced just now), so that sum is
$p(y{=}Y){=}\sum\limits_{\Upsilon} p(f(x{+}\Delta x){=}\Upsilon) p(\Delta
y{=}Y{-}\Upsilon)$.
That is $p(y{=}Y){=}\frac{1}{2^{32}}\sum\limits_{\Upsilon} p(f(x{+}{\rm
d}x){=}\Upsilon) $.
Since the sum is over all possible $\Upsilon$, the total of the summed
probabilities is 1, and $p(y{=}Y){=}1/2^{32}$. The distribution of
$\E[x_{l}]=\E[x{+}\Delta x]$ in other locations $l$ is unchanged.
At a point where two control paths join the choice of $\Delta y$ is
not free, but instead must coincide in the second path to be compiled
with the choice already made by the compiler in the first path to be compiled,
which was, however, free. If the first path does not write $l$ at
all then let an `add zero' instruction be inserted in it.
Done by structural induction on the machine code program.
\end{proof}

\noindent
That compile strategy proves Theorem~\ref{e:star} of the text, and also:

\begin{theorem}[{\bf\ref{e:ddagger}} of the main text]
Runtime user data beneath the encryption
is {\em semantically secure} against the operator
for programs compiled by the chaotic compiler of Lemma~\ref{at:4}.
\label{t:5}
\end{theorem}

\begin{proof}
Consider a probabilistic method $f$ that guesses for a particular 
runtime value beneath the encryption
`the top bit $b$ is 1, not 0', with probability $p_{C,T}$ for
program $C$ with trace $T$.
The probability that $f$ is right is
\[
\text{p}((b_{C,T}{=}1 \text{ and } f(C,T){=}1)
\text{ or }
(b_{C,T}{=}0 \text{ and } f(C,T){=}0))
\]
Splitting the conjunctions, that is
\begin{align*}
 &\text{p}(b_{C,T}{=}1) \text{p}(f(C,T){=}1\mathop|b_{C,T}{=}1)\\
+~&\text{p}(b_{C,T}{=}0) \text{p}(f(C,T){=}0\mathop|b_{C,T}{=}0)
\end{align*}
But the method $f$ cannot distinguish the compilations
it is looking at as the codes and traces are the same,
modulo the (encrypted) values in them, which the adversary cannot read.
The method $f$ applied to $C$ and $T$ has nothing to cause
it to give different answers but incidental features of encrypted
numbers and its internal spins of a coin.
Those are {\em independent} of if the bit $b$ is 1 or
0 beneath the encryption, supposing the encryption is effective. So 
\begin{align*}
\text{p}(f(C,T)=1\mathop|b_{C,T}=1) &= \text{p}(f(C,T)=1) = p_{C,T} \\
\text{p}(f(C,T)=0\mathop|b_{C,T}=0) &= \text{p}(f(C,T)=0) = 1 - p_{C,T}
\end{align*}
By Lemma~\ref{at:4}, 1 and 0 are equally likely
across all possible compilations $C$, so
the probability $f$ is right reduces to
\[
\tfrac{1}{2} ~ p_{C,T}+\tfrac{1}{2} ~ (1-p_{C,T}) = \tfrac{1}{2}
\]
since $\text{p}(b_{C,T}{=}1) = \text{p}(b_{C,T}{=}0) = \frac{1}{2}$ .
\end{proof}

The information theory in the text is based on the idea that
instructions are varied by the compiler by changing the (encrypted)
constants embedded in them, which additively varies the difference from
`nominal' of the result of the instruction through the full possible
(32-bit) range (\ref{p:2b}).  The viewpoint is that of an observer who
can see the plaintext values beneath the encryption, because an
encrypted word depends one-to-one on the plaintext when
padding is taken into account.

\def\O{\mathcal{O}}

\begin{lemma}[Lemma~\ref{l:3} of text]
The entropy of a trace is that from the instructions that
appear for a first time in it.
\end{lemma}
\begin{proof}
The delta $\Delta s$ from the nominal state $s$ in the
experienced state $s{+}\Delta s$
is what contributes to entropy, because the nominal values $s$
themselves are determined.  The delta is introduced by an
instruction $f$ that nominally has semantics $s_0
\mathop{\mapsto}\limits^f s_1$ but has been varied by the compiler to
semantics $f'$ such that $s_0{+}\Delta s_0
\mathop{\mapsto}\limits^{f'} s_1{+}\Delta s_1$ with $s_1 = f(s_0)$.  The
compiler arranges the perturbation $\Delta f$ in the constants of the
instruction so that
\begin{align*}
\Delta s_1 &= D' = f'(s_0)-f(s_0)
\end{align*}
where $D'$ is the obfuscation scheme 
at the point in the program just after the instruction, and $\Delta s_0 = D$
is that just before it. Both $D$, $D'$ are
independent of $s_0$, $s_1$.  The state change experienced
\[(s_1+\Delta s_1) - (s_0+\Delta s_0)
\]
is the part of the trace due to the instruction and, substituting, it is
\[
f(s_0)-s_0+ D' -\Delta s_0
\]
The $\Delta s_0$ is produced by previous instructions and
$f(s_0)-s_0$ is the `nominal' trace from an unperturbed
instruction semantics $f$, leaving $D'$ as the source of the
entropy contribution by this instruction.
At a second appearance in the trace of the instruction the variation
$D'$ is the same, and could have been predicted, so
contributes no entropy. 

For an instruction that reads input $x+\Delta x$ from memory
location $l$ for the first time, to register $r$, say, the data is
offset by $\Delta x = D l$ from the nominal value $x$, where $D$ is the
compiler's obfuscation scheme at the point in the program just
before the instruction runs, and the same argument applies.
\end{proof}

\section{Floating Point}

\def\stardot{\mathop{*}\limits^.}
\def\plusdot{\mathop{\text{\small+}}\limits^.}

\noindent
A practical instruction set for encrypted computing must also manipulate
floating point numbers and this is how it works.  Let single precision
floating point instructions {\bf addf}, {\bf subf}, {\bf mulf} etc.\ be
denoted by a trailing {\bf f}.  They work on registers containing
(encrypted) 32-bit integers that encode single precision floating point
numbers (`float') per the IEEE\,754 standard
\cite{cody1981analysis,Goldberg1991}.


\begin{definition}
Let $\stardot$, $\plusdot$ etc.\ denote the floating point operations
on plaintext integers encoding IEEE~754 floats.
\end{definition}
\noindent
Let
$[\,\stardot\,]$, $[\,\plusdot\,]$ etc.\  be the corresponding
operations in the ciphertext domain, in the convention noted at end
of Section~\ref{s:Not}.
The multiplication ${\bf
mulf}\,r_0\,r_1\,r_2\,k_1^\E\,k_2^\E\,k_0^\E$ has semantics
conforming to (\ref{p:2a}-\ref{p:2d}) as follows:
\begin{equation}
\tag{$\stardot$}
\label{e:dagger*f}
r_0 \leftarrow (r_1 \mathop{[-]} k^\E_1) \,\mathop{[\,\stardot\,]}\, (r_2 \mathop{[-]} k^\E_2) \mathop{[+]} k^\E_0
\end{equation}
The $-$ and $+$ are the ordinary integer subtraction and addition
operations respectively, and $[-]$ and $[+]$ are the corresponding
operations in the ciphertext domain. 
The insight is that the semantics fits the theory, and it is as `easy' for
this hardware to implement the complex IEEE operation as any other,
notwithstanding contemporary efforts to simplify and adapt floating point
(for homomorphic encryption settings), c.f.\
\cite{cheon2016floating,Costache2016,Arita2017}.  It would also not be
impractical to do in software (IEEE floating point emulation compiled
encrypted performs well in tests).

\def\eqdot{\mathop{=}\limits^.}

The branch-if-equal instruction ${\bf beqf}\,r_1\,r_2\,k_1^\E\,k_2^\E$ tests
\begin{equation}
\tag{$\eqdot$}
\label{e:dagger=f}
(r_1\mathop{[-]}k_1^\E) \,\mathop{[\eqdot]}\, (r_2\mathop{[-]} k_2^\E)
\end{equation}
where $\eqdot$  is the floating point comparison on 
floats via IEEE\,754, and $[\eqdot]$ is the corresponding
test in the ciphertext domain, with
$x^\E\mathop{[\eqdot]}y^\E \iff x\eqdot y$.  The subtraction
is as integers on the encoding, not floating point. The instruction
is atomic, per (\ref{p:2a}), with embedded encrypted constants $k_i^\E$,
$i=1,2$.

Representation of double precision 64-bit floats (`double') as
64-bit integers is also specified by IEEE\,754.  Let instructions
that manipulate those (encrypted) be denoted by a {\bf d} suffix.
Registers are referenced in pairs in the instruction by naming the first of
the pair only, its successor in the processor's register indexing scheme
being understood as the second of the pair.  The first of a pair
contains an encrypted integer representing the 32 high bits of the
IEEE\,754 encoding of a double float, the second contains an encrypted
integer encoding 32 low bits.

\begin{definition}
Let $\ddot{\hbox{\small+}}$, $\ddot{*}$, etc.\ denote
double precision floating point addition, multiplication, etc.\ on
IEEE\,754 encodings of doubles as 64-bit (i.e., $2\times32$-bit) integers.
\end{definition}
\noindent
Let $[\ddot{*}]$ be the corresponding multiplication operation in
the cipherspace domain on two pairs of encrypted 32-bit integers. Then
the multiplication instruction on encrypted double precision floats
has the following semantics, remembering registers
are referenced in pairs for double length operations, and using the
pairwise integer add and subtract operators of Definition~\ref{d:5}:
\begin{align}
 r_0 \leftarrow (r_1 \mathop{[-^2]}k_1^\E)
                    \,\mathop{[\,\ddot{*}\,]}\,
                  (r_2\mathop{[-^2]}k_2^\E) \mathop{[+^2]}
                    k_0^\E
\tag{$\ddot{*}$}
\end{align}
That satisfies (\ref{p:2a}-\ref{p:2d}).
It takes encrypted 64-bit double precision float operands in
the (pair) registers $r_1$, $r_2$  and writes to the
(pair) register $r_0$. The $k_i^\E$, $i=0,1,2$ are encrypted 64-bit
constants each embedded  as two concatenated encrypted 32-bit
constants in the instruction, which is written
$
{\bf muld}\,r_0\,r_1\,r_2\,k_0^\E\,k_1^\E\,k_2^\E
$.

\end{document}